\documentclass[12pt]{article}
\usepackage{amsmath,amsfonts,amssymb,amsthm}
\def\R{\mathbb{R}}

\newcommand{\be}{\begin{equation}}
\newcommand{\ee}{\end{equation}}
\newcommand{\bea}{\begin{eqnarray}}
\newcommand{\eea}{\end{eqnarray}}
\newcommand{\beas}{\begin{eqnarray*}}
\newcommand{\eeas}{\end{eqnarray*}}

\newtheorem{thm}{Theorem}[section]

\newtheorem{cor}[thm]{Corollary}
\newtheorem{lem}[thm]{Lemma} 
\newtheorem{prop}[thm]{Proposition}

\begin{document}

\title{Flat galaxies with dark matter halos---existence and stability}
\author{Roman Fi\v rt, Gerhard Rein, and Martin Seehafer\\
        Mathematisches Institut der
        Universit\"at Bayreuth\\
        D 95440 Bayreuth, Germany}

\maketitle
\begin{abstract}
We consider a model for a flat, disk-like galaxy surrounded by a halo of
dark matter, namely a Vlasov-Poisson type system with two particle species,
the stars which are restricted to the galactic plane and the dark matter 
particles. These constituents interact only through the gravitational potential
which stars and dark matter create collectively.  
Using a variational approach we prove the existence
of steady state solutions and their nonlinear stability under suitably
restricted perturbations.
\end{abstract}
\section{Introduction}
\setcounter{equation}{0}
Around 1970 astrophysicists noticed that in typical spiral galaxies 
the rotation velocities of the stars, when computed in the gravitational
potential of the visible matter, do not fit with their observed rotation
velocities. It was then conjectured that such galaxies are
surrounded by a halo of so far not directly observable dark matter 
in such a way
that the rotation velocities of the stars are consistent with
the resulting gravitational potential \cite{Freeman}.
For an introduction to dark matter we refer to \cite[Chapter 10]{BT} 
and the references there.

The distribution of the stars in a galaxy is usually modeled by
a density function on phase space, and it is assumed that collisions 
are sufficiently rare to be neglected and that the stars interact only
by the gravitational potential which they create collectively.
In a non-relativistic setting this results in a system of partial 
differential equations which in the mathematics literature is known as
the Vlasov-Poisson system, cf.\ \cite{rein07}. 
While the true physical nature (and existence) of dark matter are
still conjectural, we are aware of at least one astrophysics investigation
where it is also modeled as Vlasov-type matter, cf.\ \cite{RB}.
Given the fact that the only role which galactic dark matter
has to play is to provide the mass and hence the gravitational potential
needed to resolve the discrepancy concerning the rotation
velocities of the stars, such a description of dark matter seems natural.

In the present paper we investigate a model for a flat, disk-like galaxy with
a halo of dark matter where both the distribution of the stars in the galactic
plane and the distribution of the dark matter particles in the halo 
obey a Vlasov equation, and the interaction among stars, dark matter,
and between these two constituents is through the gravitational
potential which all the particles (stars and dark matter) create
collectively. 

Following the practice in astrophysics we assume that the stars are
restricted to a plane which we take to be the $x_1,x_2$ plane.
Their distribution on phase space is given by 
$\tilde f = \tilde f(t,\tilde x,\tilde v) \geq 0$ where $t\geq 0$ denotes
time and $\tilde x,\tilde v \in \R^2$ denote position and velocity in the
galactic plane. The distribution of the dark matter particles is given
by $f = f(t,x,v) \geq 0$ where 
$x,v \in \R^3$ denote position and velocity in three dimensional
space. The evolution of the galaxy and its halo is then governed by
the following Vlasov-Poisson type system of equations:
\begin{equation} \label{vlasov3d}
\partial_t f +v\cdot\nabla_{x}f-\nabla_{x} U_e \cdot\nabla_{v}f = 0,
\end{equation}
\begin{equation} \label{vlasov2d}
\partial_t \tilde f +\tilde v\cdot\nabla_{\tilde x}\tilde f
-\nabla_{\tilde x}  U_e (\cdot,0)\cdot\nabla_{\tilde v}\tilde f =0,
\end{equation}
\begin{equation} \label{potential}
 U_e(t,x) = U(t,x)+\tilde U(t,x)
= -\int_{\mathbb{R}^{3}}\frac{\rho(t,y)}{|x-y|}\,dy
-\int_{\mathbb{R}^{2}}\frac{\tilde\rho (t,\tilde y)}{|x-(\tilde y,0)|}\,
d\tilde y,
\end{equation}
\begin{equation} \label{rhodef}
\rho(t,x) 
= \int_{\mathbb{R}^{3}} f(t,x,v) dv,\
\tilde \rho(t,\tilde x) 
= \int_{\mathbb{R}^{2}} \tilde f(t,\tilde x,\tilde v) d\tilde v.
\end{equation}
Here $\rho$ and $\tilde \rho$ are the spatial mass densities of dark
matter respectively stars, $U$ and $\tilde U$ are the induced Newtonian
potentials, and $U_e$ denotes the potential of the system as
a whole, i.e., the effective potential which determines the particle orbits. 
In order that the stars remain in their plane it is sufficient to
require that 
$f(t,\tilde x, x_3,\tilde v, v_3) = f(t,\tilde x, - x_3,\tilde v,- v_3)$,
a condition which at least formally is preserved by solutions of the system
and which implies that $\nabla U(t,\tilde x,0)$ is parallel to the plane;
for $\nabla \tilde U(t,\tilde x)$ this is true automatically.
Throughout this paper we use the convention that variables with
(without) tilde denote flat (non-flat) quantities. 

To our knowledge a fully non-linear model where the gravitational
interaction within both types of matter and between the two types
is taken into account has so far not been investigated. 
Our aim is to prove the existence and non-linear stability of steady state
configurations to this system. We obtain such stable steady states
as minimizers of the total energy
\begin{eqnarray*}
&&
\frac{1}{2}\iint|v|^{2}f\,dx\,dv +
\frac{1}{2}\iint|\tilde v|^{2}\tilde f\,
d\tilde x\,d\tilde v \\
&&
\qquad
+ \frac{1}{2}\int U_e (x)\rho(x)\,dx 
+ \frac{1}{2}\int U_e(\tilde x,0)\tilde\rho(\tilde x)\,d\tilde x,
\end{eqnarray*}
satisfying suitable constraints. This so-called energy-Casimir approach
was developed for the usual, three dimensional Vlasov-Poisson system,
i.e., $\tilde f = 0$ in the above, 
in \cite{G1,G2,GR1,GR2,GR3,R6}, see also \cite{DSS,LMR,rein07,SS}. 
The approach has also been used to
prove the existence of stable steady states for flat galaxies without
a halo, i.e., with $f=0$ in the above, cf. \cite{Fi,FiR,R3}. 
The fact that in the present situation the energy is a functional
acting on two functions together with the potential interaction terms between
the flat and the non-flat component requires substantial new ingredients
in the basic scheme. One pitfall to avoid is that for a minimizer of the
above energy functional one of the two components might vanish.

Besides the above stability results it is
known that global classical solutions to the initial value problem
for the usual three dimensional Vlasov-Poisson system exist,
cf.\ \cite{LP,Pf}, while local classical and global weak
solutions exist in the flat case without halo, cf.\ \cite{Shev}.
For the situation at hand nothing is known about
the initial value problem, but we conjecture that the analogue
of \cite{Shev} for weak solutions remains true. 
Our stability result is 
conditional in the sense that it holds for solutions as long
as they exist and preserve the required conserved quantities. 
For more information on the Vlasov-Poisson system in general
we refer to the review article \cite{rein07}.

The paper proceeds as follows. In the next section we formulate
our variational problem and our main result on the existence
of minimizers. In Section~\ref{secprel} we establish properties 
of the potentials which allow us to define and control the potential energies,
in particular the interaction terms.
Next we collect some relevant results about
the decoupled variational problems where one of the two components
is missing; these facts are established in an appendix. 
In Section~\ref{secprop}
we show that the energy functional is bounded from below,
that not all the mass can escape to infinity along a minimizing sequence,
and we investigate the splitting properties of the functional.
With these prerequisites we can then prove the existence
of minimizers in Section~\ref{secexminim}. The fact that such minimizers
are steady states together with some of their properties
are established in Section~\ref{secmin=ss}. In Section~\ref{secstab}
we finally investigate the stability estimate resulting from
their minimizing property.

\section{Variational setup} \label{secvar}
\setcounter{equation}{0}
We denote the set of non-negative, Lebesgue integrable functions by 
$L^1_+(\R^n)$. For $f\in L^1_+(\R^6)$ and $\rho\in L^1_+(\R^3)$ we denote by
\[
\rho_f (x) := \int_{\mathbb{R}^{3}} f(t,x,v) dv,\
U_\rho (x) := - \int_{\mathbb{R}^{3}} \frac{\rho(y)}{|x-y|} dy
\]
the induced spatial density and gravitational potential; we write 
$U_f=U_{\rho_f}$.
Similarly, for $\tilde f\in L^1_+(\R^4)$ and $\tilde \rho\in L^1_+(\R^2)$,
\[
\rho_{\tilde f} (\tilde x) 
:= \int_{\mathbb{R}^{2}} \tilde f (t,\tilde x,\tilde v) d\tilde v,\
U_{\tilde \rho} (x) 
:= - \int_{\mathbb{R}^{2}} \frac{\tilde \rho(\tilde y)}{|x-(\tilde y,0)|}
d\tilde y,
\]
and to abbreviate we sometimes write $\tilde \rho$ and $\tilde U$ instead of
$\rho_{\tilde f}$ and $U_{\tilde \rho}$; notice that the latter
is defined on $\R^3$. In what follows we do not explicitly
denote the domain of integration---$\R^3$ or $\R^2$---unless
in cases of ambiguity. The integrability properties of these potentials
are investigated in Section~\ref{secprel}. Next we define the various parts of the
energy functional.  For  $f\in L^1_+(\R^6)$ and $\tilde f\in L^1_+(\R^4)$,
\[
E_{\mathrm{kin}}(f)
:=
\frac{1}{2}\iint|v|^{2}f(x,v)\,dv\,dx,\
E_{\mathrm{kin}}(\tilde f)
:=
\frac{1}{2}\iint|\tilde v|^{2}\tilde f(\tilde x,\tilde v)\,
d\tilde v\,d\tilde x,
\]
\[
E_{\mathrm{pot}}(f)
:=
-\frac{1}{2}\iint\frac{\rho_{f}(x)\rho_{f}(y)}{|x-y|}\,dx\,dy,\
E_{\mathrm{pot}}(\tilde f)
:=
-\frac{1}{2}\iint
\frac{\tilde\rho_{\tilde f}(\tilde x)\tilde\rho_{\tilde f}(\tilde y)}{
|\tilde x-\tilde y|}\,d\tilde x\,d\tilde y,
\]
denote the kinetic and potential energies of the non-flat and flat components.
The total energy of each component is then defined by
\[
\mathcal{H}(f) 
:= E_{\mathrm{kin}}(f) + E_{\mathrm{pot}}(f),\
\mathcal{H}(\tilde f) 
:= E_{\mathrm{kin}}(\tilde f) + E_{\mathrm{pot}}(\tilde f).
\]
Finally,
\begin{eqnarray*}
\mathcal{H}(f,\tilde f)
&=&
\mathcal{H}(f)+\mathcal{H}(\tilde f)
+\frac{1}{2}\int U_{\tilde f}(x)\rho_f(x)\,dx
+\frac{1}{2}\int U_f (\tilde x,0) \rho_{\tilde f}(\tilde x)\,d\tilde
x,\\
&=&
\mathcal{H}(f)+\mathcal{H}(\tilde f)
+ \int U_{\tilde f}(x)\rho_f(x)\,dx
\end{eqnarray*}
is the total energy of the state $(f,\tilde f)$. In Section~\ref{secprel}, where
we investigate the existence of all these integrals on the constraint set
defined below, we will also see that the two interaction terms are
equal. 

We wish to minimize this functional over the constraint set
\begin{eqnarray*}
\mathcal{F}_{\mathbf M}
:=
\Bigl\{(f,\tilde f)| 
&& 
f\in L^{1}_{+}(\mathbb{R}^{6}),\;
\tilde f\in L^{1}_{+}(\mathbb{R}^{4}),\;
||f||_{1}\le M,\;||f||_{1+1/k}\le N,\\
&&
||\tilde f||_{1}\le\tilde M,\; ||\tilde f||_{1+1/\tilde k}\le\tilde N,\
E_{\mathrm{kin}}(f) + E_{\mathrm{kin}}(\tilde f)<\infty,\\
&&f(\tilde x,x_{3},\tilde v,v_{3})=f(\tilde x,-x_{3},\tilde v,-v_{3})\Bigr\},
\end{eqnarray*}
where $\mathbf{M}:=(M,N,\tilde M,\tilde N)$ denotes the constraint vector
whose components are all strictly positive, $|| \cdot ||_p$ denotes
the usual $L^p$ norm, and
\[
0<k< 7/2,\ 0<\tilde k< 2.
\] 
In Sections~\ref{secprel} and \ref{secprop} we will see that the
total energy functional is well defined and bounded from below on this set.
The constraints on  $||f||_{1+1/k}$ and $||\tilde f||_{1+1/\tilde k}$
play the role of the Casimir constraints, and it does not seem to be possible
to include these Casimirs into the functional to be minimized, as
was done for example in \cite{GR1} for the purely three dimensional
and in \cite{R3} for the purely flat problem.
The following theorem is our main result.  
\begin{thm} \label{existence}
Let $(f_{j},\tilde f_{j})\subset\mathcal{F}_{\mathbf{M}}$ be a minimizing 
sequence of $\mathcal{H}$. Then there exists 
$(f_{0},\tilde f_{0})\in\mathcal{F}_{\mathbf{M}}$, 
a subsequence again denoted by $(f_{j},\tilde f_{j})$ and a sequence of 
shift vectors $(\tilde a_{j})\subset\mathbb{R}^{2}$ such that with 
$T_{j}f_j (x,v):=f_j (x+(\tilde a_{j},0),v)$, 
$T_{j}\tilde f_j(\tilde x,\tilde v):=
\tilde f_j(\tilde x+\tilde a_{j},\tilde v)$,
\[
T_{j}f_{j} \rightharpoonup f_{0},\ 
T_{j}\tilde f_{j} \rightharpoonup \tilde f_{0}\ 
\mbox{weakly in}\ L^{1+1/k}(\mathbb{R}^{6})\ 
\mbox{or}\ L^{1+1/\tilde k}(\mathbb{R}^{4})
\ \mbox{respectively},
\]
\[
E_{\mathrm{pot}}(T_{j}f_{j}-f_{0}) \to 0,\
E_{\mathrm{pot}}(T_{j}\tilde f_{j}-\tilde f_{0}) \to  0,
\]
and
\[
\int(\rho_{T_{j}f_{j}}-\rho_{f_{0}}) U_{T_{j}\tilde f_{j}-\tilde f_{0}}\,dx 
\to 0.
\]
Moreover $(f_{0},\tilde f_{0})$ is a minimizer of $\mathcal{H}$ over 
$\mathcal{F}_{\mathbf{M}}$.
\end{thm}
The spatial shifts parallel to the $(x_1,x_2)$ plane are
necessary due to the invariance of the total energy
and the constraint set under such shifts. If 
$(f_{0},\tilde f_{0})$ is a minimizer of $\mathcal{H}$,
then $(T_j f_0,T_j \tilde f_0)$ is a minimizing sequence for any choice
of shift vectors $\tilde a_j \in \R^2$ which is weakly convergent to a 
minimizer only if we shift our frame of reference accordingly.
\section{Preliminaries} \label{secprel}
\setcounter{equation}{0}
We start by collecting some well known estimates for the spatial densities
and potential energies induced by elements from the constraint set 
$\mathcal{F}_{\mathbf{M}}$.
\begin{lem}
\label{darkfiesta}
Let $(f,\tilde f)\in\mathcal{F}_{\mathbf{M}}$ and define
$n:=k+3/2,\ \tilde n:=\tilde k+1$. 
Then $\rho_f\in L^{1+1/n}(\mathbb{R}^{3}),\ 
\rho_{\tilde f}\in L^{1+1/\tilde n}(\mathbb{R}^{2})$ with
\begin{eqnarray*}
||\rho_f||_{1+1/n}
&\le&
C N^{(k+1)/(n+1)}E_{\mathrm{kin}}(f)^{3/(2k+5)},\\
||\rho_{\tilde f}||_{1+1/\tilde n}
&\le& 
C \tilde N^{(\tilde k+1)/(\tilde n+1)} 
E_{\mathrm{kin}}(\tilde f)^{1/({\tilde k+2})},
\end{eqnarray*}
and
\begin{eqnarray*}
-E_{\mathrm{pot}}(f)
&\le& 
C ||\rho_f||_{6/5}^{2}\le C_{\mathbf{M}}E_{\mathrm{kin}}(f)^{1/2},\\
- E_{\mathrm{pot}}(\tilde f)
&\le& 
C ||\tilde\rho_{\tilde f}||_{4/3}^{2}
\le C_{\mathbf{M}} E_{\mathrm{kin}}(\tilde f)^{1/2},
\end{eqnarray*}
where the constant $C>0$ is universal and $C_{\mathbf{M}}>0$ depends on the
constraint vector $\mathbf{M}$. By the restrictions on $k$ and 
$\tilde k$, $1+1/n > 6/5$ and $1+1/\tilde n > 4/3$ so that 
$\rho_f\in L^{6/5}(\mathbb{R}^{3})$,
$\rho_{\tilde f}\in L^{4/3}(\mathbb{R}^{2})$.
\end{lem}
\begin{proof}
Given $R>0$ we split the $v$-integral and use H\"older's inequality 
and the definition of the kinetic energy to find that
\begin{eqnarray*}
\rho_f (x)
&=&
\int_{|v|\le R}f(x,v)\,dv+\int_{|v|>R}f(x,v)\,dv\\
&\le&
\left(\frac{4\pi}{3} R^{3}\right)^{1/(k+1)}
\left(\int f^{1+1/k}(x,v)\,dv\right)^{k/(k+1)}
+\frac{1}{R^{2}}\iint|v|^{2}f(x,v)\,dv\,dx.
\end{eqnarray*}
We optimize this estimate in $R$, take the resulting estimate
for $\rho_f (x)$ to the power $1+1/n$ and integrate with respect to $x$
to obtain the estimate for $\rho_f$. The estimate for $\rho_{\tilde f}$
follows the same lines.  
The last two inequalities follow by interpolation and the 
Hardy-Littlewood-Sobolev inequality.
\end{proof}
In order to analyze the mixed term in $\mathcal{H}(f,\tilde f)$ 
we need some information on the integrability of the flat potential
in $\mathbb{R}^{3}$.
\begin{lem} \label{intflatpot}
Let $\tilde\rho\in L^{4/3}(\mathbb{R}^{2})$. 
Then $U_{\tilde \rho}\in L^{6}(\mathbb{R}^{3})$ 
and
\[
||U_{\tilde \rho}||_{L^{6}(\mathbb{R}^{3})}
\le
C ||\tilde\rho||_{L^{4/3}(\mathbb{R}^{2})}.
\]
\end{lem}
\begin{proof}
We use the general form of the Minkowski inequality, cf.\  \cite[2.4]{LL}, 
and the weak Young inequality to obtain 
\begin{eqnarray*}
||U_{\tilde \rho}||_{6}^{6}
&=&
\int_{\mathbb{R}^{2}}\int_{\mathbb{R}}
\left(\int_{\mathbb{R}^{2}}\frac{\tilde\rho(\tilde y)}
{|x-(\tilde y,0)|}\,d\tilde y\right)^{6}\,dx_{3}\,d\tilde x\\
&\le&
\int\left[\int\left(\int\frac{\tilde\rho^{6}(\tilde y)}
{(|\tilde x-\tilde y|^{2}+x_{3}^{2})^{3}}\,dx_{3}\right)^{1/6}\,
d\tilde y\right]^{6}\,d\tilde x\\
&=&
C \int\left[\int\frac{\tilde\rho(\tilde y)}{|\tilde x-\tilde y|^{5/6}}\,
d\tilde y\right]^{6}\,d\tilde x
=
C ||\tilde\rho\ast|\cdot|^{-5/6}||_{L^{6}(\mathbb{R}^{2})}^{6}\\
&\le&
C ||\tilde\rho||_{L^{4/3}(\mathbb{R}^{2})}^{6}
||\,|\cdot|^{-5/6}||_{L^{12/5}_{w}(\mathbb{R}^{2})}^{6};
\end{eqnarray*}
the function  $|\cdot|^{-\lambda}$ is an element of the weak $L^p$ space
$ L_{w}^{n/\lambda}(\mathbb{R}^{n})$, cf.\  \cite[4.3]{LL}. 
\end{proof}
We also need to investigate the integrability of $U_\rho$, restricted to the
$(x_1,x_2)$ plane.
\begin{lem} \label{3dpot2d}
Let $\rho\in L^{6/5}(\mathbb{R}^{3})$. 
Then $U_{\rho}(\cdot,0) \in L^{4}(\mathbb{R}^{2})$ with
\[
||U(\cdot,0)||_{L^{4}(\mathbb{R}^{2})} \leq
C||\rho||_{L^{6/5}(\mathbb{R}^3)}.
\]
If in addition $\tilde\rho\in L^{4/3}(\mathbb{R}^{2})$, then the 
following mixed potential energies exist and are equal:
\[
\int U_{\tilde \rho}(x)\rho (x)\,dx =
\int U_\rho (\tilde x,0) \tilde \rho(\tilde x)\,d\tilde x.
\]
\end{lem}
\begin{proof}
Fubini's theorem together with the H\"older inequality and
Lemma~\ref{intflatpot} imply that for $\rho\in L^{6/5}(\mathbb{R}^{3})$,
$\tilde\rho\in L^{4/3}(\mathbb{R}^{2})$,
\begin{eqnarray*}
\int |U_\rho (\tilde x,0)\tilde \rho(\tilde x)|\,d\tilde x
&\leq&
\iint\frac{|\rho(y) \tilde \rho (\tilde x)|}{|(\tilde x,0)-y|}\,d\tilde x\,dy
=
\int |U_{|\tilde \rho|} (y) \rho(y)|\,dy\\
&\leq&
C ||\rho||_{L^{6/5}(\mathbb{R}^3)} ||\tilde \rho||_{L^{4/3}(\mathbb{R}^2)}.
\end{eqnarray*} 
The estimate for $||U(\cdot,0)||_{L^{4}(\mathbb{R}^{2})}$ follows by
taking the supremum over all $\tilde\rho\in L^{4/3}(\mathbb{R}^{2})$
with $||\tilde\rho||_{L^{4/3}(\mathbb{R}^{2})}=1$.
Since the mixed potential energies now exist they are equal again
by Fubini's theorem.
\end{proof}

It will be useful to view the potential
energy as a bilinear form which induces a scalar product.
More precisely we define for $\rho,\sigma \in L^{6/5}(\R^3)$,
\[
\langle\rho,\sigma\rangle_\mathrm{pot} := \frac{1}{2} \int
\frac{\rho(x)\sigma(y)}{|x-y|}dy\,dx
\]
with the analogous definition for
$\langle\tilde\rho,\tilde\sigma\rangle_\mathrm{pot}$,
$\tilde\rho,\tilde\sigma \in L^{4/3}(\R^2)$, and
\begin{equation} \label{potscalprodmix}
\langle \rho,\tilde\rho\rangle_\mathrm{pot} := \frac{1}{2} \int
\frac{\rho(x)\tilde\rho(\tilde y)}{|x-(\tilde y,0)|}d\tilde
y\,dx.
\end{equation}
It is well known that $\langle\cdot,\cdot \rangle_\mathrm{pot}$
is a scalar product on $L^{6/5}(\R^3)$, cf.\ \cite[9.8]{LL},
and the same is true on $L^{4/3}(\R^2)$. The induced norms
are denoted by
\[
||\rho||_\mathrm{pot} := \langle\rho,\rho\rangle_\mathrm{pot}^{1/2},\
||\tilde \rho||_\mathrm{pot} := 
\langle\tilde \rho,\tilde \rho\rangle_\mathrm{pot}^{1/2}.
\]
Finally, 
$\langle f,g \rangle_\mathrm{pot}:=\langle \rho_f,\rho_g \rangle_\mathrm{pot}$
etc, provided that the induced spatial densities belong to 
the proper $L^p$ space, so that with this notation,
\begin{equation} \label{epotscal}
E_\mathrm{pot} (f) = -\langle f,f \rangle_\mathrm{pot} 
= - ||f||^2_\mathrm{pot}
\end{equation}
etc. The Cauchy-Schwarz
inequality corresponding to the mixed case (\ref{potscalprodmix}) 
is established next.
It tells us how strong the mixed potential energy 
term is in comparison to the potential energies of its individual components. 
\begin{lem}
\label{toll}
Let $\rho\in L_{+}^{6/5}(\mathbb{R}^{3}),\ 
\tilde\rho\in L_{+}^{4/3}(\mathbb{R}^{2})$. Then
\[
\left|\langle \rho,\tilde\rho\rangle_\mathrm{pot}\right|
\le ||\rho||_\mathrm{pot} \;
||\tilde\rho||_\mathrm{pot}.
\]
\end{lem}
\begin{proof}
We first show the assertion under the additional assumption 
that $\rho,\tilde\rho\in C^{\infty}_c$ are 
compactly supported and smooth. In that case $U_\rho$ is smooth
and bounded. Let $d\in C^\infty_c (\R^3)$ be such that $d\geq 0$ 
and $\int d =1$, and let $\delta^\epsilon(x) := \epsilon^{-3} d(x/\epsilon)$
denote the induced $\delta$-sequence; $\epsilon >0$. Then
\[
\lim_{\epsilon \to 0} \int U_\rho(\tilde x,x_3) \delta^\epsilon (x_3) dx_3 
= U_\rho(\tilde x,0)
\]
pointwise for $\tilde x \in \mathbb{R}^2$, while the latter integral 
is bounded in modulus by $||U_\rho||_\infty$. Using Lebesgue's theorem
and the fact that $\langle\cdot,\cdot \rangle_\mathrm{pot}$
is a scalar product on $L^{6/5}(\R^3)$ we can now argue as follows:
\begin{eqnarray*}
\left|\langle \rho,\tilde\rho\rangle_\mathrm{pot}\right|
&=&
\left|
\frac{1}{2} \int U_\rho (\tilde x,0) \tilde \rho(\tilde x) d\tilde x
\right|
=
\frac{1}{2} \lim_{\epsilon \to 0}
\left|
\iint U_\rho (\tilde x,x_3) \delta^\epsilon (x_3) \tilde \rho(\tilde x)\,
d x_{3}\,d\tilde x
\right|\\
&=&
\lim_{\epsilon\to 0} \left|\langle \rho,
\tilde\rho \otimes \delta^{\epsilon}\rangle_\mathrm{pot}\right|\\
&\le&
||\rho||_\mathrm{pot}\;
\lim_{\epsilon\to 0}
\left(\frac{1}{2} \iint\frac{\tilde\rho(\tilde x)\tilde\rho(\tilde y)
\delta^{\epsilon}(x_{3})\delta^{\epsilon}(y_{3})}{|x-y|}\,dx\,dy\right)^{1/2}\\
&\le& 
||\rho||_\mathrm{pot}\;
\lim_{\epsilon\to 0}
\left(\frac{1}{2} \iint\frac{\tilde\rho(\tilde x)\tilde\rho(\tilde y)
\delta^{\epsilon}(x_{3})\delta^{\epsilon}(y_{3})}{|\tilde x-\tilde y|}
\,dx\,dy\right)^{1/2}\\
&=&
||\rho||_\mathrm{pot}\;
||\tilde\rho||_\mathrm{pot};
\end{eqnarray*}
in the last step we used that 
$\int \delta^\epsilon = 1$ for $\epsilon >0$.
The general case follows by approximating $\rho$ and $\tilde \rho$
by compactly supported, smooth functions, observing the fact that
both sides of the inequality are continuous with respect to the
$L^{6/5}(\R^3)$-norm for $\rho$ and  the $L^{4/3}(\R^2)$-norm
for $\tilde \rho$.  
\end{proof}
\section{The decoupled minimizers} \label{secdecoup}
\setcounter{equation}{0}
In the next sections the existence and properties of the minimizers of the
decoupled problems where one of the components is missing will become
important. 
Here we briefly collect the relevant facts.
A function $g$ on $\R^d \times \R^d$ is called spherically symmetric
if for every $A \in \mathrm{SO}(d)$, $g(Ax,Av)=g(x,v)$.

For each $M, N >0$
the energy $\mathcal{H}(f)$ has a minimizer $f_0^{\mathrm{3D}}$ in the set
\[
\mathcal{F}_{M,N}^\mathrm{3D}
:=
\Bigl\{ f\in L^{1}_{+}(\mathbb{R}^{6}) \mid\,
||f||_{1}\le M,\;||f||_{1+1/k}\le N,\
E_{\mathrm{kin}}(f) <\infty \Bigr\}.
\]
The minimizer is unique up
to spatial shifts, spherically symmetric, has negative energy, i.e., 
$\mathcal{H}(f_0^{\mathrm{3D}}) < 0$,
saturates the constraints, i.e.,
$||f_0^{\mathrm{3D}}||_{1} = M$, $||f_0^{\mathrm{3D}}||_{1+1/k} =N$,
and has compact spatial support. 
There exists a constant $R^\ast >0$ which is independent of $M$ and $N$
such that the radius of this spatial support is
\begin{equation} \label{radius3d}
R = R^\ast M^{(2k-1)/3} N^{-(2k+2)/3}.
\end{equation}
By spherical symmetry, 
\[
f_0^{\mathrm{3D}}(\tilde x,x_{3},\tilde v,v_{3})=
f_0^{\mathrm{3D}}(\tilde x,-x_{3},\tilde v,-v_{3}).
\]
Similarly, the energy $\mathcal{H}(\tilde f)$ has a minimizer $f_0^{\mathrm{FL}}$ 
in the set 
\[
\mathcal{F}_{M,N}^\mathrm{FL}
:=
\Bigl\{ \tilde f\in L^{1}_{+}(\mathbb{R}^{4}) \mid
||\tilde f||_{1}\le\tilde M,\; ||\tilde f||_{1+1/\tilde k}\le\tilde N,\
E_{\mathrm{kin}}(\tilde f)<\infty \Bigr\}.
\]
A slight complication arises from the fact that we do
at the moment not know whether this minimizer is again unique up
to spatial shifts. However, there does exist a
two-parameter family $(f_{M,N}^{\mathrm{FL}})_{M,N >0}$
such that $f_{M,N}^{\mathrm{FL}}$
is a minimizer of  $\mathcal{H}(\tilde f)$ over
$\mathcal{F}_{M,N}^\mathrm{FL}$
which saturates the constraints, has negative energy,
is axially symmetric with respect to the $x_3$ axis, 
i.e., spherically symmetric as a function of $\tilde x, \tilde v$,
and has compact spatial support. 
There exists a constant $\tilde R^\ast$ independent of $M$ 
and $N$
such that the radius of this spatial support is
\begin{equation} \label{radius2d}
\tilde R = \tilde R^\ast M^{\tilde k} N^{-(\tilde k+1)}.
\end{equation}
In what follows $f_0^{\mathrm{FL}}$ always denotes the corresponding
member of the above family. In particular, if 
\[
\mathbf{M} = (M,N,\tilde M,\tilde N) =:
\left(\mathbf{M}^\mathrm{3D},\mathbf{M}^\mathrm{FL}\right)
\]
then $f_0^\mathrm{3D}$ denotes the minimizer of $\mathcal{H}$
over $\mathcal{F}_{M,N}^\mathrm{3D}$ and $f_0^{\mathrm{FL}}$
denotes $f_{\tilde M,\tilde N}^{\mathrm{FL}}$.

Since the above facts are known or follow by arguments
already available in the literature we defer their discussion
to the appendix.
\section{Properties of $\mathcal{H}$}\label{secprop}
\setcounter{equation}{0}
First we establish a lower bound for $\mathcal{H}$ on
$\mathcal{F}_{\mathbf{M}}$ and certain a-priori bounds 
along minimizing sequences.
\begin{lem} \label{lowerbound}
\begin{itemize}
\item[(a)]
The functional $\mathcal{H}$ is bounded from below on 
$\mathcal{F}_{\mathbf{M}}$, i.e.,
\[
-\infty < \inf_{\mathcal{F}_{\mathbf{M}}}\mathcal{H} =: h_{\mathbf{M}} < 0.
\]
\item[(b)]
Along every minimizing sequence 
$(f_{j},\tilde f_{j})\subset\mathcal{F}_{\mathbf{M}}$ of
$\mathcal{H}$
both the kinetic and the potential energies are bounded, more
precisely, for $j$ sufficiently large,
\[
E_{\mathrm{kin}}(f_{j})+E_{\mathrm{kin}}(\tilde f_{j})
+|E_{\mathrm{pot}}(f_{j})|+|E_{\mathrm{pot}}(\tilde f_{j})|\le C_{\mathbf{M}},
\]
where the constant $C_{\mathbf{M}}>0$ depends only on $\mathbf{M}$.
\end{itemize}
\end{lem}
\begin{proof}
Lemma~\ref{darkfiesta} and Lemma~\ref{toll} imply that 
for $(f,\tilde f)\in\mathcal{F}_{\mathbf{M}}$, 
\begin{eqnarray*}
\left|\langle f,\tilde f\rangle_\mathrm{pot}\right|
&\le &
||f||_\mathrm{pot} ||\tilde f||_\mathrm{pot}
\le  
C_{\mathbf{M}} E_{\mathrm{kin}}(f)^{1/4} 
E_{\mathrm{kin}}(\tilde f)^{1/4}\\
&\le&
C_{\mathbf{M}} E_{\mathrm{kin}}(f)^{1/2} 
+ C_{\mathbf{M}}E_{\mathrm{kin}}(\tilde f)^{1/2}.
\end{eqnarray*}
Using Lemma~\ref{darkfiesta} again this yields the estimate
\begin{equation} \label{hlower}
\mathcal{H}(f,\tilde f)
\ge
E_{\mathrm{kin}}(f)-C_{\mathbf{M}}E_{\mathrm{kin}}(f)^{1/2}
+ E_{\mathrm{kin}}(\tilde f)-C_{\mathbf{M}} E_{\mathrm{kin}}(\tilde f)^{1/2}. 
\end{equation}
Hence $h_{\mathbf{M}} > -\infty$. Moreover,
\[
h_{\mathbf{M}}
\leq \mathcal{H}(f_{0}^{\mathrm{3D}},f_{0}^{\mathrm{FL}})
=\mathcal{H}(f_{0}^{\mathrm{3D}})+\mathcal{H}(f_{0}^{\mathrm{FL}})
+\int\tilde U_{0}\rho_{0}\,dx < 0.
\]
Hence along a minimizing sequence $\mathcal{H}(f_j,\tilde f_j)\leq 0$ for $j$
sufficiently large, and by (\ref{hlower}),
\[
\left(E_{\mathrm{kin}}(f_j)^{1/2} - C_{\mathbf{M}}/2\right)^2 +
\left(E_{\mathrm{kin}}(\tilde f_j)^{1/2} - C_{\mathbf{M}}/2\right)^2
\leq C_{\mathbf{M}}^2 /2.
\]
Another reference to Lemma~\ref{darkfiesta} completes the proof.
\end{proof}
In order to pass to the limit along a minimizing sequence we need
the following compactness properties of the potential energies;
by $\mathbf{1}_S$ we denote the indicator function of the set $S$,
and we recall (\ref{epotscal}) and the corresponding notation.
\begin{lem}
\label{compact}
Let $(\rho_{j})\subset L^{1+1/n}(\mathbb{R}^{3})$ and 
$(\tilde\rho_{j})\subset L^{1+1/\tilde n}(\mathbb{R}^{2})$ be such that
\[
\rho_{j}\rightharpoonup\rho_{0}\ 
\mbox{weakly in}\ L^{1+1/n}(\mathbb{R}^{3}),\
\tilde\rho_{j}\rightharpoonup\tilde\rho_{0}\ 
\mbox{weakly in}\ L^{1+1/\tilde n}(\mathbb{R}^{2}).
\]
Then for each $R>0$,
\[
||\mathbf{1}_{B_{R}}(\rho_{j}-\rho_{0})||_{\mathrm{pot}} \to 0,\
||\mathbf{1}_{\tilde B_{R}}(\tilde\rho_{j}-\tilde\rho_{0})||_{\mathrm{pot}} 
\to 0 \ \mbox{as}\ j\to \infty.
\]
\end{lem}
\begin{proof}
The convergence of the non-flat potential energy is proved for example in
\cite[Lemma 2.5]{rein07}. 
For the flat case we refer to \cite[Lemma 3.6]{FiR}.
\end{proof}
A crucial step in the analysis is to show that minimizing sequences
do not spread out in space and that up to spatial shifts
not all the mass can leak out to infinity. This is the content of the
next result.
\begin{prop}  \label{konc}
Let $(f_{j},\tilde f_{j})\subset\mathcal{F}_{\mathbf{M}}$ be a 
minimizing sequence of $\mathcal{H}$.
Then there exists a sequence $(\tilde a_{j})\subset\mathbb{R}^{2}$ 
of shift vectors, $\epsilon_{0}>0$, and $R_{0}>0$ such that
for all sufficiently large $j\in\mathbb{N}$,
\[
\int_{(\tilde a_{j},0)+B_{R_{0}}}\int f_{j}\,dv\,dx \ge \epsilon_{0},\quad
\int_{\tilde a_{j}+\tilde B_{R_{0}}}\int \tilde f_{j}\,d\tilde v\,d\tilde x
\ge\epsilon_{0}.
\]
Here $B_{R_{0}}$ and $\tilde B_{R_{0}}$ denote the closed ball of radius
$R_0$ about the origin in $\R^3$ or $\R^2$ respectively.
\end{prop}
\noindent
{\bf Remark.} It is important that the same shift vectors
work for both the non-flat and the flat component.
\begin{proof}
Let $U_j := U_{ f_j}$, $\tilde\rho_j := \rho_{\tilde f_j}$, 
and let $R^{\mathrm{3D}}$ and $R^{\mathrm{FL}}$ denote the radii of the
decoupled minimizers $f_0^{\mathrm{3D}}$ and $f_0^{\mathrm{FL}}$ 
subject to constraints $\mathbf{M}^{\mathrm{3D}}$ and
$\mathbf{M}^{\mathrm{FL}}$, cf.\ Section~\ref{secdecoup}.
Since
\[
\lim_{j\to \infty} \mathcal{H}(f_j,\tilde f_j) \leq 
\mathcal{H}(f_0^{\mathrm{3D}},f_0^{\mathrm{FL}})
=\mathcal{H}(f_0^{\mathrm{3D}}) +  \mathcal{H}(f_0^{\mathrm{FL}}) 
+ \int U_0^{\mathrm{3D}} \rho_0^{\mathrm{FL}} d\tilde x,
\]
we get that for $j$ sufficiently large,
\[
\mathcal{H}(f_j) + \mathcal{H}(\tilde f_j) + \int  U_j\tilde\rho_j d\tilde x <
\mathcal{H}(f_0^{\mathrm{3D}}) +  \mathcal{H}(f_0^{\mathrm{FL}}) 
+ \frac{1}{2}\int U_0^{\mathrm{3D}} \rho_0^{\mathrm{FL}} d\tilde x.
\]
Since $\mathcal{H}(f_j) \geq \mathcal{H}(f_0^{\mathrm{3D}})$
and $\mathcal{H}(\tilde f_j) \geq \mathcal{H}(f_0^{\mathrm{FL}})$
this implies that
\begin{eqnarray}
\label{negene3}
\int  U_j \tilde\rho_j d\tilde x  
&\leq& 
\frac{1}{2}\int  U_0^{\mathrm{3D}}\rho_0^{\mathrm{FL}}d\tilde x
= -\frac{1}{2} \iint\frac{\rho_0^{\mathrm{3D}}(y)\rho_0^{\mathrm{FL}}(\tilde x)}
{|(\tilde x,0)-y|} d\tilde x\, dy \nonumber\\ 
&<&
 -\frac{M\tilde M}{2(R^{\mathrm{3D}}+R^{\mathrm{FL}})}
\end{eqnarray}
for all sufficiently large $j\in\mathbb{N}$.

For $R>1$ we write
\beas
\frac{1}{|x|}
&=& 
\mathbf{1}_{\{|x| \leq 1/R\}}(x)\frac{1}{|x|} 
+ \mathbf{1}_{\{1/R<|x| < R\}}(x)\frac{1}{|x|}
+ \mathbf{1}_{\{|x| \geq R\}}(x)\frac{1}{|x|} \\
&=:& 
K^1_R(x) + K^2_R(x) + K^3_R(x).
\eeas
With this splitting 
\[
\left|\int  U_j\tilde\rho_j d\tilde x\right|  
= \iint\frac{\rho_j(y)\tilde \rho_j(\tilde x)}{|(\tilde x,0)-y|} d\tilde x\,dy 
= J_1+J_2+J_3.
\]
The second and third terms are estimated straightforwardly:
\beas
J_2  
&\leq& 
R \iint_{|(\tilde x,0)-y|<R}\rho_j(y)
\tilde \rho_j(\tilde x)\,d\tilde x\,dy, \\
J_3  
&\leq& 
R^{-1}\iint\rho_j(y)\tilde \rho_j(\tilde x)\,d\tilde x\,dy 
\leq M\tilde M R^{-1}.
\eeas
For $J_1$ we first apply the H\"older inequality and then
the general form of the Minkowski inequality as in the proof of
Lemma~\ref{intflatpot} to obtain the estimate
\beas
J_1 
&\leq& 
\|\rho_j\|_{1+1/n}
\left\|\int_{|\tilde x-\tilde \cdot|<1/R} 
\frac{\tilde \rho_j(\tilde x)}{|(\tilde x,0)-\cdot|} d\tilde x \right\|_{n+1} \\
&\leq& 
C\, \|\rho_j\|_{1+1/n} \|\tilde \rho_j * (\tilde K^1_R)^{n/(n+1)}\|_{n+1} \\
&\leq& 
C\, |\rho_j\|_{1+1/n} \|\tilde \rho_j\|_{1+1/\tilde n} 
\|(\tilde K^1_R)^{n/(n+1)}\|_\gamma \\
&\leq&
C\, \|\rho_j\|_{1+1/n} \|\tilde \rho_j\|_{1+1/\tilde n}R^{-\sigma} \leq
C_{\mathbf{M}}R^{-\sigma},
\eeas
where
\[
\gamma := \left(\frac{1}{n+1} + \frac{1}{\tilde n+1}\right)^{-1}>1,\quad
\sigma := \frac{2}{\gamma} - \frac{n}{n+1}>0;
\]
recall that $3/2 < n < 5$ and $1<\tilde n<3$.
With (\ref{negene3}) we find that
\[
-\frac{M\tilde M}{2(R^{\mathrm{3D}}+R^\mathrm{FL})} > 
\int  U_j\tilde\rho_j d\tilde x = -J_1 - J_2 - J_3,
\]
and hence
\[
J_2 \geq
\frac{M\tilde M}{2(R^{\mathrm{3D}}+R^\mathrm{FL})} 
- \frac{M\tilde M}{R} - C_{\mathbf{M}}R^{-\sigma}.
\]
For $R$ sufficiently large the right hand side is positive, so that
\be \label{concentration}
0 < R^{-1}\left(\frac{M\tilde M}{2(R^{\mathrm{3D}}+R^\mathrm{FL})} 
- \frac{M\tilde M}{R} - C_{\mathbf{M}}R^{-\sigma} \right)
\leq  \iint_{|x-(\tilde y,0)|<R}\rho_j(x)\tilde 
\rho_j(\tilde y)\,dx\,d\tilde y. 
\ee
The existence of the shift vectors $(\tilde a_j)$ 
with the asserted properties is now a consequence of the following lemma. 
\end{proof}
\begin{lem}
\label{problem}
Let $\rho \in L^1_+ (\R^3)$, $\sigma \in L^1_+(\R^2)$ with 
\[ 
0< \int_{\R^3} \rho(x) dx, \int_{\R^2} \sigma(\tilde y)\, d\tilde y 
\leq M < \infty
\]
and such that for some $\delta_0,\, r_0>0$
\[
\iint_{|x-(\tilde y,0)|<r_0}\rho(x)\sigma(\tilde y)\,dx\,d\tilde y > 
\delta_0.
\]
Then there exist $\epsilon_0,\,R_0> 0$ depending only on 
$\delta_0$, $r_0$, and $M$ such that
\[
\epsilon_0 < \int_{|x-(\tilde a,0)|<R_0}\rho(x) dx\ \mbox{and}\ 
\epsilon_0 < \int_{|\tilde y-\tilde a|<R_0}\sigma(\tilde y) d\tilde y 
\]
for some $\tilde a \in \R^2$.
\end{lem}
\begin{proof}
Let $z\in \R^3$ be given. Note first that
\beas
&&
\left\{(x,\tilde y) \in \R^5 \mid |z-(\tilde y,0)|<r_0,\ 
|x-(\tilde y,0)|<r_0\right\} \\
&&
\qquad \qquad \qquad \qquad \subset
\left \{(x,\tilde y) \in \R^5 \mid |z-x|<2 r_0,\ |x-(\tilde y,0)|<r_0 \right\}
\eeas
and hence
\beas
&&
\int_{|z-(\tilde y,0)|<r_0}\left(\int_{|x-(\tilde y,0)|<r_0} 
\sigma(\tilde y)dx\right)d\tilde y \\
&&
\qquad \qquad \qquad \qquad \leq 
\int_{|z-x|<2r_0}\left(\int_{|x-(\tilde y,0)|<r_0}
\sigma(\tilde y)d\tilde y\right)dx.
\eeas
Multiplying with $\rho(z)$ and integrating with respect 
to $z \in \R^3$ we obtain
\bea \label{central_ineq}
&&
\int\rho(z)\int_{|z-(\tilde y,0)|<r_0}
\left(\int_{|x-(\tilde y,0)|<r_0} \sigma(\tilde y)dx\right)\,d\tilde y\, dz 
\nonumber \\
&&
\qquad \qquad  \qquad \leq
\int\rho(z)\int_{|z-x|<2r_0}\left(\int_{|x-(\tilde y,0)|<r_0}
\sigma(\tilde y)\,d\tilde y\right)\,dx\, dz.
\eea
Changing the order of integration, the right hand side of 
(\ref{central_ineq}) can be rewritten as
\beas
&&
\int_{x \in \R^3} \int_{|z-x|<2r_0} \rho(z) 
\left(\int_{|x-(\tilde y,0)|<r_0}
\sigma(\tilde y)\,d\tilde y\right)\,dz\,dx \\
&&
\qquad \qquad 
= \int_{x \in \R^3}\left(\int_{|z-x|<2r_0}
\rho(z)\,dz\right)\,\left(\int_{|x-(\tilde y,0)|<r_0}
\sigma(\tilde y)\,d\tilde y\right)\,dx \\
&&
\qquad \qquad
= \int_{\R^3} R(x)\,S(x)\,dx,
\eeas
where the functions $R$ and $S$ are defined by
\be \label{def_RS}
R(x) :=  \int_{|x-z|<2r_0}\rho(z)\, dz, \quad
S(x) :=  \int_{|x-(\tilde y,0)|<r_0}\sigma(\tilde y)\, d\tilde y.
\ee
From our hypothesis,
\begin{eqnarray*}
\delta_0 
&<&
\int_{z \in \R^3}\rho(z)\int_{|z-(\tilde y,0)|<r_0} 
\sigma(\tilde y)\,d\tilde y\, dz\\
&=& 
\frac{3}{4\pi r_0^3} 
\int_{z \in \R^3} \rho(z) \int_{|z-(\tilde y,0)|<r_0}
\int_{|x-(\tilde y,0)|<r_0}\sigma(\tilde y)\,dx\,d\tilde y\,dz,
\end{eqnarray*}
so that combining with (\ref{central_ineq}) we are led to
\begin{equation} \label{snd_ineq}
\delta_0 < \frac{3}{4\pi r_0^3}\int_{\R^3} R(x)\,S(x)\,dx.
\end{equation}
As a direct consequence of our definitions (\ref{def_RS}) we find that
\begin{equation}
\|R\|_\infty \leq M, \quad \|S\|_\infty \leq M. \label{infty_bound} 
\end{equation}
Furthermore,
\begin{eqnarray*}
\|R\|_1 
&=& 
\int_{x\in \R^3}\int_{|z-x|<2r_0}\rho(z)\,dz\,dx
 = \int_{z\in \R^3}\int_{|z-x|<2r_0}\rho(z\,)dx\,dz \\
&\leq& 
8M\frac{4\pi}{3}r_0^3,
\end{eqnarray*}
and 
\begin{eqnarray*}
\|S\|_1 
&=& 
\int_{x \in \R^3}\int_{|x-(\tilde y,0)|<r_0}\sigma(\tilde y)\,d\tilde y\,dx
 = \int_{\tilde y \in \R^2}\int_{|x-(\tilde y,0)|<r_0}
\sigma(\tilde y)\,dx\,d\tilde y \\
&\leq & 
M\frac{4\pi}{3}r_0^3.
\end{eqnarray*}
We may thus continue with (\ref{snd_ineq}) as follows:
\begin{eqnarray*}
\delta_0 
& < & 
\frac{3}{4\pi r_0^3} \int_{\R^3} R(x)\,S(x)\,dx \\
&\leq& 
\frac{3}{4\pi r_0^3}\|(R\,S)^{1/2}\|_\infty 
\int_{\R^3}(R(x)\,S(x))^{1/2}dx \\
&\leq& 
\frac{3}{4\pi r_0^3}\|R\,S\|^{1/2}_\infty 
\left(\int R(x)\,dx\right)^{1/2}\left(\int S(x)\,dx\right)^{1/2} \\
&\leq& 
2\sqrt{2}M \|R\,S\|^{1/2}_\infty.
\end{eqnarray*}
So there exists $a \in \R^3$ such that
\[
R(a)\,S(a) > \left(\frac{\delta_0}{2\sqrt{2}M}\right)^2.
\]
In view of (\ref{infty_bound}) this implies that
\[
R(a) > \frac{\delta_0^2}{8 M^3}\ \mbox{and}\ 
S(a) > \frac{\delta_0^2}{8 M^3}.
\]
Finally we write $a=(\tilde a,a_3)$ with $\tilde a \in \R^2$, $a_3 \in \R$ 
and observe that
\[
\int_{|\tilde a-\tilde y|<r_0}\sigma(\tilde y)\,d\tilde y \geq 
\int_{|(\tilde a,a_3)-(\tilde y,0)|<r_0}\sigma(\tilde y)\,d\tilde y 
=S(a) >  \frac{\delta_0^2}{8 M^3}.
\]
In addition, $S(a)>0$ clearly implies $|a_3|<r_0$, so that
\[
\int_{|z-(\tilde a,0)|<3r_0}\rho(z)\,dz 
\geq \int_{|a-z|<2r_0}\rho(z)\,dz = R(a) > \frac{\delta_0^2}{8 M^3},
\]
which is exactly our claim with 
\be \label{epsrel}
\epsilon_0 := \frac{\delta_0^2}{8 M^3},\ R_0 := 3 r_0.
\ee
\end{proof}
In what follows it is important to control the
parameters $\epsilon_{0}$ and $R_{0}$ in Proposition~\ref{konc}
if the constraint vector $\mathbf{M}$ varies. This
is the content of the following corollary.
\begin{cor}
\label{coro}
Let the constraint vector $\mathbf{M}$ satisfy the bounds
\[
0<M_l\le M\le M_u,\ 0<\tilde M_l\le \tilde M\le \tilde M_u,
\]
\[
0<N_l\le N\le N_u,\ 0<\tilde N_l\le \tilde N\le \tilde N_u.
\] 
Then the parameters $\epsilon_{0}$ and $R_{0}$ in Proposition~\ref{konc}
can be chosen independently of $\mathbf{M}$ and $(f_{j},\tilde f_{j})$, 
depending only on the bounds $\mathbf{M}_l$ and $\mathbf{M}_u$.
\end{cor}
\begin{proof}
Under the given bounds on $\mathbf{M}$ we can choose $R>0$ depending only on
these bounds such that
the left hand side in
(\ref{concentration}) is bounded from below by a parameter
$\delta_0 >0$ also depending only on these bounds. To this end,
observe that $M$ and $\tilde M$ are bounded both from below and above,
$C_\mathbf{M}$ is bounded from above, and the radii $R^\mathrm{3D}$
and $R^\mathrm{FL}$ are bounded from above in view of
(\ref{radius3d}) and (\ref{radius2d}). Given (\ref{epsrel})
this completes the proof.
\end{proof}
The last tool needed for the proof of Theorem~\ref{existence}
is the fact that the energy infimum $h_\mathbf{M}$ is sub-additive
in $\mathbf{M}$. While up to now all components of the constraint
vector $\mathbf{M}$ were strictly positive, this sub-additivity
is for technical reasons needed
also in situations where the flat or the non-flat component
of a constraint vector vanishes, i.e., 
$\mathbf{M} = (\mathbf{M}^\mathrm{3D},\mathbf{M}^\mathrm{FL})$
and $\mathbf{M}^\mathrm{FL}=0$ or $\mathbf{M}^\mathrm{3D}=0$. In such a case 
$h_\mathbf{M}$ is obviously  taken to denote
$\mathcal{H}(f_0^\mathrm{3D})$ or
$\mathcal{H}(f_0^\mathrm{FL})$ respectively, where $f_0^\mathrm{3D}$ is the
minimizer of $\mathcal{H}$ over $\mathcal{F}^\mathrm{3D}_{\mathbf{M}^\mathrm{3D}}$
and $f_0^\mathrm{FL}$ is the
one over $\mathcal{F}^\mathrm{FL}_{\mathbf{M}^\mathrm{FL}}$, 
cf.\ Section~\ref{secdecoup}.
We say that
{\em the constraint vector $\mathbf{M}\in [0,\infty[^{4}$ 
is nontrivial}, if
\[
(M> 0\land N> 0)\lor(\tilde M> 0\land\tilde N> 0).
\]
\begin{prop} \label{subadlem}
For all $\mathbf{M}_{1},\, \mathbf{M}_{2}\in [0,\infty[^{4}$, 
\[
h_{\mathbf{M}_{1}+\mathbf{M}_{2}} \leq
h_{\mathbf{M_{1}}}+h_{\mathbf{M}_{2}}.
\]
If both $\mathbf{M}_{1}$ and $\mathbf{M}_{2}$ are 
nontrivial, then this inequality is strict.
If $\mathbf{M}_{1}$ 
satisfies uniform bounds from above and below as in Corollary~\ref{coro}, 
and if either this is true also for $\mathbf{M}_{2}$
or one component of $\mathbf{M}_{2}$ vanishes and the
other one satisfies such uniform bounds,
then there exists $\epsilon>0$ depending  
only on these bounds such that
\[
h_{\mathbf{M}_{1}+\mathbf{M}_{2}} \leq
h_{\mathbf{M_{1}}}+h_{\mathbf{M}_{2}} - \epsilon.
\]
\end{prop}
\begin{proof}
Consider two minimizing sequences 
$(f^{1}_{j},\tilde f^{1}_{j})\subset \mathcal{F}_{\mathbf{M}_1}$
and
$(f^{2}_{j},\tilde f^{2}_{j})\subset \mathcal{F}_{\mathbf{M}_2}$
with
\[
\mathcal{H}(f^{1}_{j},\tilde f^{1}_{j}) \to h_{\mathbf{M}_{1}},\quad
\mathcal{H}(f^{2}_{j},\tilde f^{2}_{j}) \to h_{\mathbf{M}_{2}}.
\]
If one of the constraint vectors, say $\mathbf{M}_{2}$ has a trivial
component, say the flat one, then we take $(f_0^\mathrm{3D},0)$ as
the corresponding minimizing sequence.
By the Minkowski inequality, 
$(f^{1}_{j}+f^{2}_{j},\tilde f^{1}_{j}+\tilde f^{2}_{j})
\in\mathcal{F}_{\mathbf{M}_{1}+\mathbf{M}_{2}}$, 
and hence
\begin{eqnarray*}
h_{\mathbf{M}_{1}+\mathbf{M}_{2}}
&\le&
\mathcal{H}(f^{1}_{j}+f^{2}_{j},\tilde f^{1}_{j}+\tilde f^{2}_{j})\\
&=&
\mathcal{H}(f^{1}_{j},\tilde f^{1}_{j})+
\mathcal{H}(f^{2}_{j},\tilde f^{2}_{j})\\
&&
{}- 2 \langle f^{1}_{j}, f^{2}_{j}\rangle_\mathrm{pot}
-2 \langle \tilde f^{1}_{j},\tilde f^{2}_{j}\rangle_\mathrm{pot}
-2 \langle f^{1}_{j},\tilde f^{2}_{j}\rangle_\mathrm{pot}
-2 \langle f^{2}_{j},\tilde f^{1}_{j}\rangle_\mathrm{pot}\\
&\le&
\mathcal{H}(f^{1}_{j},\tilde f^{1}_{j})+\mathcal{H}(f^{2}_{j},\tilde
f^{2}_{j})
\to h_{\mathbf{M}_1} +h_{\mathbf{M}_2}.
\end{eqnarray*}
If $\mathbf{M}_{1}$ and $\mathbf{M}_{2}$ have both at least one 
nontrivial component, then the corresponding potential interaction
energy is strictly negative so that the estimate above is strict.

Assume now that we have positive uniform lower and upper bounds for
$\mathbf{M}_{1}$ and $\mathbf{M}_{2}$.
We can assume that the minimizing sequences are  shifted in
such a way that the assertions of Proposition~\ref{konc}  
hold with
$\epsilon_{0}^{1}$, $\epsilon_{0}^{2}$, $R_{0}^{1}$, and $R_{0}^{2}$, 
and without spatial shifts. If one component of $\mathbf{M}_{2}$
vanishes, say, the flat one, the corresponding trivial
minimizing sequence $(f_0^\mathrm{3D},0)$ need of course not be shifted,
and we take for $R_{0}^{2}$ the radius of the minimizer $f_0^\mathrm{3D}$
and for $\epsilon_{0}^{2}$ its mass.
In either case
\[
\langle \rho^{1}_{j}, \rho^{2}_{j}\rangle_\mathrm{pot}
\geq
\iint_{B_{R_0^1}\times B_{R_0^2}}
\frac{\rho^{1}_{j}(x)\rho^{2}_{j}(y)}{|x-y|}\,dx\,dy\\
\geq
\frac{\epsilon^{1}_{0}\epsilon^{2}_{0}}{R^{1}_{0}+R^{2}_{0}}.
\]
The latter quantity is now bounded from below
by some $\epsilon >0$ depending only on the uniform
bounds on the constraint vectors, where
we use Corollary~\ref{coro} and in addition (\ref{radius3d})
if the flat component of $\mathbf{M}_{2}$ vanishes.
If the non-flat component of $\mathbf{M}_{2}$ vanishes
we use (\ref{radius2d}) instead.
\end{proof}

\smallskip

\noindent
{\bf Remark.} The uniform sub-additivity is also valid
if both $\mathbf{M}_{1}$ and $\mathbf{M}_{2}$ have exactly one nontrivial
component which is uniformly bounded from below and above,
but this case is not needed in what follows.
\section{Proof of Theorem~\ref{existence}} \label{secexminim}
\setcounter{equation}{0}
Let $(f_{j},\tilde f_{j})\in \mathcal{F}_{\mathbf{M}}$ be a
minimizing sequence for $\mathcal{H}$. We choose shift 
vectors $\tilde a_{j}\in \R^2$ such that the assertion 
of Proposition~\ref{konc} holds. 
To keep the notation simple we redefine $(f_{j},\tilde f_{j})$
as the minimizing sequence shifted by these vectors as in the statement
of Theorem~\ref{existence}. Hence according to Proposition~\ref{konc},
\begin{equation} \label{concj}
\epsilon_{0}\leq \int_{|x|\leq R_0}\int f_{j}\,dv\,dx \leq M, \quad
\epsilon_{0}\leq \int_{|\tilde x|\leq R_0}\int \tilde f_{j}\,d\tilde v\,d\tilde x
\leq \tilde M.
\end{equation}
This new sequence is of course minimizing as well. 
The definition of $\mathcal{F}_{\mathbf{M}}$ implies the a-priori bounds
\[
||f_{j}||_{1+1/k}\le N,\quad
||\tilde f_{j}||_{1+1/\tilde k}\le\tilde N.
\]
Hence after extracting a subsequence which we denote by the same symbol,
\[
f_{j}\rightharpoonup f_{0}\ \mbox{weakly in}\ L^{1+1/k}(\mathbb{R}^{6}),\quad
\tilde f_{j}\rightharpoonup \tilde f_{0}\ \mbox{weakly in}\ 
L^{1+1/\tilde k}(\mathbb{R}^{4}).
\]
From this weak convergence it follows that
\[
||f_{0}||_{1} \le M,\
||\tilde f_{0}||_{1} \le \tilde M,\
||f_{0}||_{1+1/k}\le N,\
||\tilde f_{0}||_{1+1/\tilde k}\le\tilde N,
\]
and
\[
E_{\mathrm{kin}}(f_{0}) \le
\limsup_{j\to\infty}E_{\mathrm{kin}}(f_{j})<\infty,\quad
E_{\mathrm{kin}}(\tilde f_{0}) \le
\limsup_{j\to\infty} E_{\mathrm{kin}}(\tilde f_{j})<\infty.
\]
By Lemma \ref{darkfiesta} the corresponding spatial densities 
$\rho_{j}:=\rho_{f_{j}}$ and $\tilde\rho_{j}:=\rho_{\tilde f_{j}}$ 
are bounded in $L^{1+1/n}(\mathbb{R}^{3})$ or 
$L^{1+1/\tilde n}(\mathbb{R}^{2})$ respectively.
After extracting a subsequence again,
\[
\rho_{j} \rightharpoonup \rho_{0}\ \mbox{weakly in}\  
L^{1+1/n}(\mathbb{R}^{3}),\quad
\tilde \rho_{j} \rightharpoonup \tilde \rho_{0}\ \mbox{weakly in}\ 
L^{1+1/\tilde n}(\mathbb{R}^{2}).
\]
It is easy to see that
in fact $\rho_{0}=\rho_{f_{0}}$ and 
$\tilde\rho_{0}=\rho_{\tilde f_{0}}$. 
The essential step is to prove that up to extracting yet another 
subsequence the potential energy terms converge, i.e.,
\begin{eqnarray*} 
||f_{j}-f_{0}||_{\mathrm{pot}}+
||\tilde f_{j}-\tilde f_{0}||_{\mathrm{pot}}
\to 0\ \mbox{as}\ j\to\infty;
\end{eqnarray*}
by Lemma~\ref{toll} it then follows that also
$\langle f_{j}-f_{0}, 
\tilde f_{j}-\tilde f_{0}\rangle_{\mathrm{pot}}\to 0$. 

For $R > R_1\geq R_0$ we define
$B_{R_1,R}:=\{x\in\mathbb{R}^{3} \mid R_1\le|x| < R\}$ 
with the obvious definition of $\tilde B_{R_1,R}$, 
and we split the functions $f_{j}$ and $\tilde f_{j}$ as follows:
\begin{eqnarray*}
f_{j}
&=&
\mathbf{1}_{B_{R_1}\times\mathbb{R}^{3}} f_{j}+
\mathbf{1}_{B_{R_1,R}\times\mathbb{R}^{3}} f_{j}+
\mathbf{1}_{B_{R,\infty}\times\mathbb{R}^{3}}f_{j} 
=: 
f_{j}^{1}+f_{j}^{2}+f_{j}^{3},\\
\tilde f_{j}
&=&
\mathbf{1}_{\tilde B_{R_1}\times\mathbb{R}^{2}} \tilde f_{j} +
\mathbf{1}_{\tilde B_{R_1,R}\times\mathbb{R}^{2}} \tilde f_{j} +
\mathbf{1}_{\tilde B_{R,\infty}\times\mathbb{R}^{2}} \tilde f_{j}
=:
\tilde f_{j}^{1}+\tilde f_{j}^{2}+\tilde f_{j}^{3}.
\end{eqnarray*}
Lemma~\ref{compact} implies that for $R > R_1\geq R_0$ fixed,
\begin{equation} \label{epot2conv}
||f_{j}^{1}+f_{j}^{2}-f_{0}^{1}-f_{0}^{2}||_{\mathrm{pot}}
+
||\tilde f_{j}^{1}+\tilde f_{j}^{2}-
\tilde f_{0}^{1}-\tilde f_{0}^{2}||_{\mathrm{pot}}
\to 0\ \mbox{as}\ j\to\infty.
\end{equation}
So we only need to show that for any $\epsilon>0$ and $R$ 
sufficiently large,
\begin{equation}
\label{hlspor}
\liminf_{j\to\infty}
\left(||f_{j}^{3}||_{\mathrm{pot}} +
||\tilde f_{j}^{3}||_{\mathrm{pot}} \right)<\epsilon.
\end{equation}
Once this is established we use the triangle
inequality for $||\cdot||_{\mathrm{pot}}$
to conclude that  
\[
||f_{j}-f_{0}||_{\mathrm{pot}}
\le
||f_{j}^{1}+f_{j}^{2}-f_{0}^{1}-f_{0}^{2}||_{\mathrm{pot}}+
||f_{j}^{3}||_{\mathrm{pot}}+ ||f_{0}^{3}||_{\mathrm{pot}}.
\]
We can surely find $R>1$ such that the right hand side is as small as we 
want for $j$ sufficiently large. Hence for $j\to \infty$, 
\[
E_{\mathrm{pot}}(f_{j})\to E_{\mathrm{pot}}(f_{0})
\]
and with the same argument,
\[
E_{\mathrm{pot}}(\tilde f_{j})\to E_{\mathrm{pot}}(\tilde f_{0}).
\]
Finally by Lemma \ref{toll},
\[
\int \tilde U_{j} \rho_{j}\,dx\to
\int \tilde U_{0}\rho_{0}\,dx,
\]
and all together implies that
\[
\mathcal{H}(f_{0},\tilde f_{0})\le
\lim_{j\to\infty}\mathcal{H}(f_{j},\tilde f_{j})
= h_\mathbf{M}.
\]
This is the desired minimizing property of $(f_{0},\tilde f_{0})$.

We prove (\ref{hlspor}) by contradiction, so assume that
(\ref{hlspor}) is false, i.e.
\[
\exists\epsilon_1>0\,\forall R>1\,\exists j_{0}\in\mathbb{N}\,
\forall j\ge j_{0}:\
||f_{j}^{3}||_{\mathrm{pot}} +||\tilde f_{j}^{3}||_{\mathrm{pot}}
\ge\epsilon_1.
\]
Then we can choose a subsequence such that without change of labeling 
it satisfies either
\be \label{f3large}
\forall R>1\,\exists j_{0}\in\mathbb{N}\,
\forall j\ge j_{0}:\ 
||f_{j}^{3}||_{\mathrm{pot}}\ge\epsilon_1/2
\ee
or
\be \label{tildef3large}
\forall R>1\,\exists j_{0}\in\mathbb{N}\,
\forall j\ge j_{0}:\ 
||\tilde f_{j}^{3}||_{\mathrm{pot}}\ge\epsilon_1/2.
\ee
In the following we consider the first case, 
the second one can be treated analogously.
The contradiction is arrived at by splitting $f_j$ and $\tilde f_j$
as above and then using the uniform sub-additivity
from Proposition~\ref{subadlem}. Let us denote
\[
f^0_j:=\mathbf{1}_{B_{R_0}\times\mathbb{R}^{3}} f_{j},\
\tilde f^0_j:=\mathbf{1}_{\tilde B_{R_0}\times\mathbb{R}^{2}} \tilde f_{j}.
\]
Since the splitting parameters satisfy the relation
$R>R_1\geq R_0$, (\ref{concj}) implies that 
\be \label{f1l1bounds}
\epsilon_0 \leq ||f^0_j||_1 \leq  ||f^1_j||_1 \leq M,\ 
\epsilon_0 \leq ||\tilde f^0_j||_1 \leq ||\tilde f^1_j||_1 \leq \tilde M.
\ee
We also need uniform lower bounds for the $L^{1+1/k}$-norm and
$L^{1+1/\tilde k}$-norm.
By Lemma~\ref{darkfiesta},
\[
\|f^0_j\|_1 
=
\|\rho^0_j\|_1 \leq C(R_0)\, \|\rho^0_j\|_{1+1/n} 
\leq C(R_0)\, \|f^0_j\|_{1+1/k} ^{(k+1)/(n+1)},
\]
with an analogous estimate for $\tilde f_j$.
Hence with (\ref{concj}),
\be \label{f1lpbounds}
0 < C(\epsilon_0) \leq ||f^1_j||_{1+1/k} \le N,\
0 < C(\epsilon_0)  \leq 
||\tilde f_{j}^1||_{1+1/\tilde k}\le \tilde N.
\ee
From the assumption (\ref{f3large}) we now derive such bounds
also for $f_j^3$.
By Lemma~\ref{darkfiesta} with $\theta\in]0,1[$ an interpolation
parameter and $\sigma := (1-\theta)(1+k)/(1+n)$,
\be \label{pee1}
||f||_{\mathrm{pot}}^2 = |E_{\mathrm{pot}}(f)|
\le
C\,||\rho||_{6/5}^{2} 
\le C\,||\rho||_{1}^{2\theta} ||\rho||_{1+1/n}^{2(1-\theta)}
\le C\,||f||_{1}^{2\theta}||f||_{1+1/k}^{2 \sigma}.
\ee
With $f = f_j^3$ this implies that
\be \label{f3bounds}
0 < C(\epsilon_1) \le ||f_{j}^{3}||_{1} \le M,\
0 < C(\epsilon_1) \le ||f_{j}^{3}||_{1+1/k} \le N.
\ee
To arrive at a contradiction we insert the splitting of
$f_j$ and $\tilde f_j$ into the energy functional:
\begin{eqnarray}
\mathcal{H}(f_{j},\tilde f_{j})
&=&
\mathcal{H}(f_{j}^{1},\tilde f_{j}^{1})+\mathcal{H}(f_{j}^{2},\tilde f_{j}^{2})
+\mathcal{H}(f_{j}^{3},\tilde f_{j}^{3})\nonumber \\
&&
- 2 \langle  f_{j}^{2},  f_{j}^{1}+ f_{j}^{3}\rangle_{\mathrm{pot}}  
- 2 \langle  f_{j}^{1},  f_{j}^{3}\rangle_{\mathrm{pot}}
- 2 \langle \tilde  f_{j}^{2},
\tilde  f_{j}^{1}+\tilde f_{j}^{3}\rangle_{\mathrm{pot}}
- 2 \langle\tilde  f_{j}^{1},\tilde  f_{j}^{3}\rangle_{\mathrm{pot}}\nonumber\\
&&
- 2 \langle f_{j}^{2},\tilde
 f_{j}^{1}+\tilde f_{j}^{3}\rangle_{\mathrm{pot}}
- 2 \langle f_{j}^{1},\tilde f_{j}^{3}\rangle_{\mathrm{pot}}
- 2 \langle f_{j}^{1}+ f_{j}^{3},\tilde  f_{j}^{2}\rangle_{\mathrm{pot}}
- 2 \langle f_{j}^{3},\tilde f_{j}^{1}\rangle_{\mathrm{pot}}\nonumber\\
&=:&
\mathcal{H}(f_{j}^{1},\tilde f_{j}^{1})+
\mathcal{H}(f_{j}^{2},\tilde f_{j}^{2})+
\mathcal{H}(f_{j}^{3},\tilde f_{j}^{3})\nonumber\\
&&-I_{1}-I_{2}-\tilde I_{1}-\tilde I_{2}-J_{1}-
J_{2}-\tilde J_{1}-\tilde J_{2}. \label{splitt}
\end{eqnarray}
Using the Cauchy-Schwarz inequality for 
$\langle \cdot,\cdot\rangle_{\mathrm{pot}}$,
i.e., Lemma~\ref{toll} for the mixed terms
$J_1$ and $\tilde J_1$,
and the boundedness of potential energies along the 
minimizing sequence, cf.\ Lemma~\ref{lowerbound}, 
we obtain the estimates
\begin{eqnarray*}
I_{1} + J_{1}
&\le &
C\,||f_{j}^{2}||_{\mathrm{pot}}
\le C\,\left(||f_{j}^{2}-f_{0}^{2}||_{\mathrm{pot}}+
||f_{0}^{2}||_{\mathrm{pot}}\right),\\
\tilde I_{1}+ \tilde J_{1}
&\le &
C\,||\tilde f_{j}^{2}||_{\mathrm{pot}}
\le C\,\left(||\tilde f_{j}^{2}-\tilde f_{0}^{2}||_{\mathrm{pot}}+
||\tilde f_{0}^{2}||_{\mathrm{pot}}\right).
\end{eqnarray*}
For $R>2 R_1$ and $x\in B_{R_1},\ y\in B_{R,\infty}$ we note that
\[
\frac{1}{|x-y|}\le\frac{1}{|y|-R_1}\le\frac{1}{|y|-|y|/2}=\frac{2}{|y|},
\]
and we combine this with the H\" older inequality to estimate 
$I_{2}$, $\tilde I_{2}$, $J_{2}$, and $\tilde J_{2}$ as follows:
\begin{eqnarray*}
I_{2}
&\le &
2\int_{B_{R_1}}\rho_{j}(x)\,dx
\int_{B_{R,\infty}}|y|^{-1}\rho_{j}(y)\,dy\le C||\rho_{j}||_{6/5}^{2}
\left(\frac{R_1}{R}\right)^{1/2},\\
\tilde I_{2}
&\le &
2\int_{\tilde B_{R_1}}\tilde \rho_{j}(\tilde x)\,d\tilde x
\int_{\tilde B_{R,\infty}}|\tilde y|^{-1}\tilde \rho_{j}(\tilde y)\,d\tilde y
\le C||\tilde \rho_{j}||_{4/3}^{2}\left(\frac{R_1}{R}\right)^{1/2},\\
J_{2}
&\le &
2\int_{B_{R_1}}\rho_{j}(x)\,dx
\int_{\tilde B_{R,\infty}}|\tilde y|^{-1}\tilde \rho_{j}(\tilde y)\,d\tilde y
\le C||\rho_{j}||_{6/5}||\tilde\rho_{j}||_{4/3}
\left(\frac{R_1}{R}\right)^{1/2},\\
\tilde J_{2}
&\le &
2\int_{\tilde B_{R_1}}\tilde \rho_{j}(\tilde x)\,d\tilde x
\int_{B_{R,\infty}}|y|^{-1}\rho_{j}(y)\,dy\le C||\rho_{j}||_{6/5}
||\tilde \rho_{j}||_{4/3}\left(\frac{R_1}{R}\right)^{1/2}.
\end{eqnarray*}
We wish to apply the uniform sub-additivity from Proposition~\ref{subadlem}
to the constraint vectors induced by 
$(f^1_j,\tilde f^1_j)$ and $(f^3_j,\tilde f^3_j)$.
To this end, let
\begin{eqnarray*}
\mathbf{M}_j 
&:=& 
(\|f_j\|_1, \|f_j\|_{1+1/k},\|\tilde f_j\|_1, 
\|\tilde f_j\|_{1+1/\tilde k} ), \\
 \mathbf{M}^i_j 
&:=& 
(\|f^i_j\|_1, \|f^i_j\|_{1+1/k},\|\tilde f^i_j\|_1, 
\|\tilde f^i_j\|_{1+1/\tilde k}), \ i=1,2,3. 
\end{eqnarray*}
From (\ref{f1l1bounds}), (\ref{f1lpbounds}), (\ref{f3bounds})
we have the required uniform bounds for $\mathbf{M}^1_j$
and $(M^3_j,N^3_j)$. With respect to $\tilde f^3_j$ we now
distinguish two cases. Either this function also satisfies such
non-zero uniform bounds or it is negligible, more precisely:

\smallskip

\textit{Case 1:} $\exists\, \epsilon_2>0\,\forall\, R>1\,
\exists\, j_0 \in \mathbb{N}\, \forall j\geq j_0 :\
|E_{\mathrm{pot}}(\tilde f^3_j)| \geq \epsilon_2$.\\
In this case the analogue of the potential energy estimate (\ref{pee1}) 
for $\tilde f^3_j$ implies that
\[
0 < C(\epsilon_2) \le ||\tilde f_{j}^{3}||_{1} \le \tilde M,\
0 < C(\epsilon_2) \le ||\tilde f_{j}^{3}||_{1+1/\tilde k}\le \tilde N.
\]
So we have obtained uniform positive bounds for each entry of the quantities 
$\mathbf{M}^1_j$ and $\mathbf{M}^3_j$ from above and below.
By Proposition~\ref{subadlem},
\[
\mathcal{H}(f_{j}^{1},\tilde f_{j}^{1})+
\mathcal{H}(f_{j}^{2},\tilde f_{j}^{2})+
\mathcal{H}(f_{j}^{3},\tilde f_{j}^{3}) \ge
h_{\mathbf{M}^{1}_{j}}+h_{\mathbf{M}^{2}_{j}}+h_{\mathbf{M}^{3}_{j}}
\ge h_{\mathbf{M}}+\epsilon
\]
with $\epsilon >0$ independent of the splitting parameters $R > 2 R_1$ and 
of $j$.
Recalling (\ref{splitt}) we find that
\begin{eqnarray*}
&&
h_{\mathbf{M}}-\mathcal{H}(f_{j},\tilde f_{j})+\epsilon
\le I_{1}+I_{2}+\tilde I_{1}+\tilde I_{2}+J_{1}+J_{2}+
\tilde J_{1}+\tilde J_{2}\\
&&
\qquad\le C_{1}\Bigl[||f^{2}_{0}||_{\mathrm{pot}} +
||\tilde f^{2}_{0}||_{\mathrm{pot}} + 
||f^{2}_{j}-f^{2}_{0}||_{\mathrm{pot}} +
||\tilde f^{2}_{j}-\tilde f^{2}_{0}||_{\mathrm{pot}}
+\left(R_1/R\right)^{1/2}\Bigr].
\end{eqnarray*}
We choose $R_1\geq R_0$ such that 
\[
C_{1}(||f^{2}_{0}||_{\mathrm{pot}}+||\tilde f^{2}_{0}||_{\mathrm{pot}})
<\epsilon/4.
\]
Next we choose $R>2 R_1$ such that $C_{1}(R/R_1)^{1/2}\le\epsilon/4$.
For $j$ large,
\[
h_{\mathbf{M}}-\mathcal{H}(f_{j},\tilde f_{j})+\epsilon
\le
\frac{1}{2}\epsilon
+C_{1}\left[ ||f^{2}_{j}-f^{2}_{0}||_{\mathrm{pot}} +
||\tilde f^{2}_{j}-\tilde f^{2}_{0}||_{\mathrm{pot}}\right],
\]
and by (\ref{epot2conv}) this contradicts the fact that 
$(f_{j},\tilde f_{j})$ is minimizing.

\smallskip

\textit{Case 2:} $\forall\, \epsilon>0\,
\exists\, R_\epsilon >1\,\forall\, j_0 \in \mathbb{N}\,
\exists\, j\geq j_0:\
||\tilde f^3_j||_{\mathrm{pot}} < \epsilon$, provided $R\geq R_\epsilon$.\\
In this case we neglect $\tilde f^3_j$ in the sub-additivity argument
and recall that Proposition~\ref{subadlem} yields
$\epsilon_2>0$ only depending 
on the bounds
for $\mathbf{M}_j^1$, $M_j^3$, and $N_j^3$ such that 
\beas
h_{\mathbf{M}^1_j} + h_{\mathbf{M}^2_j} + h_{(M^3_j,N^3_j,0,0)} 
&\geq&
h_{\mathbf{M}^1_j + (M^3_j,N^3_j,0,0)} +\epsilon_2+ h_{\mathbf{M}^2_j}\\
&\geq&
h_{\mathbf{M}^1_j + (M^3_j,N^3_j,0,0) + \mathbf{M}^2_j} + \epsilon_2\\
&\geq&
h_{\mathbf{M}^1_j + (M^3_j,N^3_j,0,0) + \mathbf{M}^2_j} + \epsilon_2
+ h_{(0,0,\tilde M^3_j,\tilde N^3_j)} \geq h_{\mathbf{M}_j}+ \epsilon_2.
\eeas
By the assumption of the present case we can choose a subsequence which we 
keep on denoting as before such that 
$||\tilde f^3_j||_{\mathrm{pot}} < \epsilon$
for all $j\in\mathbb{N}$, where $\epsilon$ will be determined in terms of
$\epsilon_2$ below; if necessary we increase $R$ so that
$R\geq R_\epsilon$. By Lemma~\ref{toll},
\[
|\langle \tilde f^3_j,f^3_j\rangle_{\mathrm{pot}}| \leq C \epsilon.
\]
Moreover, we choose $R>2 R_1\geq 2 R_0$ such that in (\ref{splitt}), 
$I_1 + \cdots +\tilde J_2 < \epsilon$. Hence by (\ref{splitt}),
\beas
\mathcal{H}(f_j,\tilde f_j) 
&\geq&
\mathcal{H}(f^1_j,\tilde f^1_j)+\mathcal{H}(f^2_j,\tilde f^2_j)
+\mathcal{H}(f^3_j,\tilde f^3_j) - \epsilon \\
&\geq&
\mathcal{H}(f^1_j,\tilde f^1_j)+\mathcal{H}(f^2_j,\tilde f^2_j)
+\mathcal{H}(f^3_j) - \epsilon^2 - C \epsilon\\ 
&\geq&
h_{\mathbf{M}^1_j} + h_{\mathbf{M}^2_j} + h_{(M^3_j,N^3_j,0,0)} 
- \epsilon^2 - C \epsilon\\
&\geq&
h_{\mathbf{M}_j}+ \epsilon_2 - \epsilon^2 - C \epsilon.
\eeas
If $\epsilon$ is chosen properly in terms of $\epsilon_2$, 
\[
\mathcal{H}(f_j,\tilde f_j) 
\geq h_{\mathbf{M}} + \epsilon_2/2 \ \mbox{as}\ j\to \infty.
\]
This contradicts the minimizing property of
$(f_j,\tilde f_j)$.
If one considers the case (\ref{tildef3large}) instead of (\ref{f3large}),
all the arguments remain the same with the roles of the flat and non-flat
components interchanged.
The proof of Theorem~\ref{existence} is complete.
\section{Properties of the minimizer} \label{secmin=ss}
\setcounter{equation}{0}
First we exclude the possibility that for a minimizer 
$f_{0}= 0$ or $\tilde f_{0}= 0$. 
Indeed the next result shows that the constraints are 
to some extent  saturated by any minimizer.
\begin{prop} \label{saturate} 
Let $(f_{0},\tilde f_{0})\in\mathcal{F}_{\mathbf{M}}$ be a minimizer 
of $\mathcal{H}$ over $\mathcal{F}_{\mathbf{M}}$. Then
\[
||f_{0}||_{1}=M \; \lor \; ||\tilde f_{0}||_{1}=\tilde M,
\]
\[
||f_{0}||_{1+1/k}=N,\
||\tilde f_{0}||_{1+1/\tilde k} =\tilde N.
\]
\end{prop}
\begin{proof}
We define for $a,b,c,d,e>0$ 
a rescaled state $(f_{0}^{\ast},\tilde f_{0}^{\ast})$ as
\[
f_{0}^{\ast}(x,v):=af_{0}(bx,cv),\quad
\tilde f_{0}^{\ast}(\tilde x,\tilde v):=d\tilde f_{0}(b\tilde x,e\tilde v);
\]
because of the mixed potential energy term $x$ and $\tilde x$ must be
scaled in the same way.
Then
\[
E_{\mathrm{kin}}(f_{0}^{\ast})
=
ab^{-3}c^{-5}E_{\mathrm{kin}}(f_{0}),\
E_{\mathrm{kin}}(\tilde f_{0}^{\ast})
=
db^{-2}e^{-4}E_{\mathrm{kin}}(\tilde f_{0}),
\]
\[
E_{\mathrm{pot}}(f_{0}^{\ast})
=
a^{2}b^{-5}c^{-6}E_{\mathrm{pot}}(f_{0}),\
E_{\mathrm{pot}}(\tilde f_{0}^{\ast})
=
d^{2}b^{-3}e^{-4}E_{\mathrm{pot}}(\tilde f_{0}),
\]
\[
\int \tilde U_{0}^{\ast}\rho_{0}^{\ast}\,dx
=
adb^{-4}c^{-3}e^{-2}\int \tilde U_{0}\rho_{0}\,dx.
\]
Assume that $||f_{0}||_{1+1/k}<N$. Then we choose 
\[
a=c^{3},\ b=d=e=1.
\]
For this choice of parameters $\tilde f_0^\ast = \tilde f_0$,
\[
||f_{0}^{\ast}||_{1} = ||f_{0}||_{1},\
||f_{0}^{\ast}||_{1+1/k} = c^{3/(k+1)}||f_{0}||_{1+1/k},
\]
and
\[
\mathcal{H}(f_{0}^{\ast},\tilde f_{0}^{\ast})
=c^{-2}E_{\mathrm{kin}}(f_{0})+E_{\mathrm{pot}}(f_{0})
+\mathcal{H}(\tilde f_{0})+\int \tilde U_{0} \rho_{0}\,dx.
\]
We can choose $c>1$ so that the rescaled state 
still lies in $\mathcal{F}_{\mathbf{M}}$ and has lower 
energy which is a contradiction. 
The analogous argument shows that  
$||\tilde f_{0}||_{1+1/\tilde k}=\tilde N$. 

In order to prove that at least one of the two mass constraints
is saturated we assume 
that $||f_{0}||_{1}<M\land ||\tilde f_{0}||_{1}<\tilde M$, 
and we choose the scaling parameters 
\[
a=c^{-7},\ b=d=c^{-4},\ e=c.
\]
For this choice,
\[
||f_{0}^{\ast}||_{1}
=
c^{2}||f_{0}||_{1},\
||\tilde f_{0}^{\ast}||_{1}
=
c^{2}||\tilde f_{0}||_{1},
\]
\[
||f_{0}^{\ast}||_{1+1/k}
=
c^{(2k-7)/(k+1)}||f_{0}||_{1+1/k},\
||\tilde f_{0}^{\ast}||_{1+1/\tilde k}
=
c^{(2\tilde k-4)/(\tilde k+1)}||\tilde f_{0}||_{1+1/\tilde k},
\]
and
\[
\mathcal{H}(f_{0}^{\ast},\tilde f_{0}^{\ast})=\mathcal{H}(f_{0},\tilde f_{0}).
\]
Since $0<k<7/2$ and $0<\tilde k<2$ we can choose $c>1$ such that 
$(f_{0}^{\ast},\tilde f_{0}^{\ast})\in\mathcal{F}_{\mathbf{M}}$ and 
both 
$\mathbf{M^\ast}:= (||f_0^\ast||_1,||f_0^\ast||_{1+1/k},
||\tilde f_0^\ast||_1,||\tilde f_0^\ast||_{1+1/\tilde k})$
and $\mathbf{M}-\mathbf{M^\ast}$ are non-trivial. 
The strict sub-additivity in Proposition~\ref{subadlem} implies the
desired contradiction:
\[
h_{\mathbf{M}}<h_{\mathbf{M}^\ast}+h_{\mathbf{M}-\mathbf{M}^\ast}
< \mathcal{H}(f_{0}^{\ast},\tilde f_{0}^{\ast})
= \mathcal{H}(f_{0},\tilde f_{0}) = h_{\mathbf{M}}.
\]
\end{proof}
The main result of this section is the fact that
the minimizers are 
functions of the particle or local energy. We use the
Lagrange multiplier method presented for example in \cite{GR1,GR2,rein07,SS}. 
\begin{thm}
\label{minim=ss}
Let $(f_{0},\tilde f_{0})$ be a minimizer as obtained 
in Theorem~\ref{existence}
with induced potentials $(U_{0},\tilde U_{0})$. Then 
\[
f_{0}(x,v)=\left(\frac{E_{0}-E(x,v)}{\lambda}\right)_{+}^{k}
\ \mbox{a.e.},
\]
\[
\tilde f_{0}(\tilde x,\tilde v)
=\left(\frac{\tilde E_{0}-E(\tilde x,0,\tilde v,0)}
{\tilde\lambda}\right)_{+}^{\tilde k}\ \mbox{a.e.},
\]
where $E(x,v):=\frac{1}{2}|v|^{2}+U_{0}(x)+\tilde U_{0}(x)$
and $(\cdot)_+$ denotes the positive part. 
The Lagrange multipliers are defined as
\begin{eqnarray*}
E_{0}
&:=&
\frac{1}{||f_{0}||_{1}}\left(\frac{2k+5}{3}
E_{\mathrm{kin}}(f_{0})+2E_{\mathrm{pot}}(f_{0})+
\int U_{0} \tilde\rho_{0} \,dx\right),\\
\tilde E_{0}
&:=&
\frac{1}{||\tilde f_{0}||_{1}}
\left((\tilde k+2)E_{\mathrm{kin}}(\tilde f_{0})+
2E_{\mathrm{pot}}(\tilde f_{0})+
\int U_{0} \tilde\rho_{0}\,dx\right),
\end{eqnarray*}
and
\[
\lambda
:=
\frac{2(k+1)E_{\mathrm{kin}}(f_{0})}{3||f_{0}||_{1+1/k}^{1+1/k}},\
\tilde\lambda
:=
\frac{(\tilde k+1)E_{\mathrm{kin}}(\tilde f_{0})}
{||\tilde f_{0}||_{1+1/\tilde k}^{1+1/\tilde k}}.
\]
\end{thm}
\begin{proof}
Let $(f_{0},\tilde f_{0})$ be a minimizer of 
$\mathcal{H}$ with corresponding potentials 
$(U_{0},\tilde U_{0})$. For $f$ such that 
$(f,\tilde f_0)\in \mathcal{F}_\mathbf{M}$ we define
\[
\mathcal{G}(f):=\mathcal{H}(f,\tilde f_{0}).
\]
Then
\begin{eqnarray}
\label{distG}
\mathcal{G}(f)-\mathcal{G}(f_{0})
&=&
E_{\mathrm{kin}}(f)-
E_{\mathrm{kin}}(f_{0})+E_{\mathrm{pot}}(f)-E_{\mathrm{pot}}(f_{0})
\nonumber \\
&&
{}+
\int(\rho_{f}-\rho_{0})\tilde U_{0}\,dx.
\end{eqnarray}
For each fixed $\epsilon>0$ we define the set 
\[
S_{\epsilon}:=\left\{(x,y)\in\mathbb{R}^{6} \mid
\epsilon\le f_{0}(x,v)\le\epsilon^{-1}\right\}.
\]
Let $\eta\in L^{\infty}(\mathbb{R}^{6})$ be a real-valued function 
with compact support such that $\eta\ge 0$ a.e. for 
$(x,v)\in\mathbb{R}^{6}\setminus\mathrm{supp}\,f_{0}$ and 
$\mathrm{supp}\,\eta \subset (\mathbb{R}^{6}\setminus
\mathrm{supp}\,f_{0})\cup S_{\epsilon}$. 
For $t\in[0,T]$ and $T=(||\eta||_{1}+||\eta||_{1+1/k}+
||\eta||_{\infty})^{-1}\epsilon/2$ we define
\[
f_{t}(x,v):=\alpha^{3}(t)||f_{0}||_{1}
\frac{f_{0}+t\eta}{||f_{0}+t\eta||_{1}}(x,\alpha(t)v),
\]
where
\[
\alpha(t):=\left(\frac{||f_{0}||_{1+1/k}}{||f_{0}||_{1}}
\frac{||f_{0}+t\eta||_{1}}{||f_{0}+t\eta||_{1+1/k}}\right)^{(k+1)/3}.
\]
For $t\in[0,T]$,
\[
||f_{t}||_{1}=||f_{0}||_{1},\quad ||f_{t}||_{1+1/k}=||f_{0}||_{1+1/k}
\]
and $f_{0}+t\eta\ge 0$ a.e. For $\epsilon$ small enough,
\[
\frac{||f_{0}||_{1}}{2}\le||f_{0}+t\eta||_{1}\le ||f_{0}||_{1}
+\frac{\epsilon}{2},
\]
\[
\frac{||f_{0}||_{1+1/k{1}}}{2}
\le||f_{0}+t\eta||_{1+1/k}\le ||f_{0}||_{1+1/k}+\frac{\epsilon}{2},
\]
which implies that $\alpha$ is a smooth function on $[0,T]$ and
\[
\alpha'(t)=\frac{k+1}{3}\alpha(t)
\left[\frac{||\eta||_{1}}{||f_{0}+t\eta||_{1}}-
\frac{\iint(f_{0}+t\eta)^{1/k}\eta\,dx\,dv}{||f_{0}+
t\eta||_{1+1/k}^{1+1/k}}\right].
\]
Moreover, $\alpha''$ is bounded on $[0,T]$. From 
(\ref{distG}) we conclude that for $t\in[0,T]$,
\begin{eqnarray}
\mathcal{G}(f_{t})-\mathcal{G}(f_{0})
&=&
\left(\frac{||f_{0}||_{1}}{\alpha^{2}(t)||f_{0}+t\eta||_{1}}-1\right)
E_{\mathrm{kin}}(f_{0})
+\frac{||f_{0}||_{1}t}{\alpha^{2}(t)||f_{0}+t\eta||_{1}}E_{\mathrm{kin}}(\eta)
\nonumber \\
&&
{}+\left(\frac{||f_{0}||_{1}^{2}}{||f_{0}+t\eta||_{1}^{2}}-1\right)
E_{\mathrm{pot}}(f_{0})+\frac{||f_{0}||_{1}^{2}t}{||f_{0}+t\eta||_{1}^{2}}
\int\rho_{\eta}U_{0}\,dx\nonumber\\
&&
{}+\frac{||f_{0}||_{1}^{2}t^{2}}{||f_{0}+t\eta||_{1}^{2}}
E_{\mathrm{pot}}(\eta)+
\left(\frac{||f_{0}||_{1}}{||f_{0}+t\eta||_{1}}-1\right)
\int\rho_{0}\tilde U_{0}\,dx\nonumber\\
&&
{}+\frac{||f_{0}||_{1}t}{||f_{0}+t\eta||_{1}}\int\rho_{\eta}\tilde U_{0}\,dx.
\label{distG2}
\end{eqnarray}
By Taylor expansion at $t=0$,
\begin{eqnarray*}
\frac{||f_{0}||_{1}}{\alpha^{2}(t)||f_{0}+t\eta||_{1}}-1
&=&
-t\left[\frac{||\eta||_{1}}{||f_{0}||_{1}}+
2\frac{k+1}{3}\left(\frac{||\eta||_{1}}{||f_{0}||_{1}}-
\frac{\iint f_{0}^{1/k}\eta\,dx\,dv}
{||f_{0}||_{1+1/k}^{1+1/k}}\right)\right]+\mathrm{O}(t^{2}),\\
\frac{||f_{0}||_{1}t}{\alpha^{2}(t)||f_{0}+t\eta||_{1}}
&=&
t+\mathrm{O}(t^{2}),\\
\frac{||f_{0}||_{1}^{2}}{||f_{0}+t\eta||_{1}^{2}}-1
&=&
-\frac{2||\eta||_{1}t}{||f_{0}||_{1}}+\mathrm{O}(t^{2}),\\
\frac{||f_{0}||_{1}^{2}t}{||f_{0}+t\eta||_{1}^{2}}
&=&
t+\mathrm{O}(t^{2}),\\
\frac{||f_{0}||_{1}}{||f_{0}+t\eta||_{1}}-1
&=&
-\frac{||\eta||_{1}t}{||f_{0}||_{1}}+\mathrm{O}(t^{2}),\\
\frac{||f_{0}||_{1}t}{||f_{0}+t\eta||_{1}}
&=&
t+\mathrm{O}(t^{2}).
\end{eqnarray*}
If we substitute these expansions into (\ref{distG2}), we find that
\[
\mathcal{G}(f_{t})-\mathcal{G}(f_{0})
=t\iint(E-E_{0}+ \lambda f_{0}^{1/k})\eta\,dv\,dx+\mathrm{O}(t^{2})
\]
with $E_{0}$ and $\lambda$ as given in the theorem. 
Since $\mathcal{G}(f_{t})$ attains its minimum at $t=0$,
the choice of $\eta$ and  
$\epsilon\to 0$ imply that $E-E_{0}\ge 0$ on 
$\mathbb{R}^{6}\setminus\mathrm{supp}\,f_{0}$ and
\[
f_{0}=\left(\frac{E_{0}-E}{\lambda}\right)^{k}
\ \mbox{a.e.\ on}\ \mathrm{supp}\,f_{0}.
\]
If we repeat this argument with the roles of flat and non-flat
states exchanged, i.e., for 
$\mathcal{G}(\tilde f):=\mathcal{H}(f_0,\tilde f)$,
we obtain the assertion for $\tilde f_0$.
\end{proof}

The previous theorem states that for a minimizer $(f_0,\tilde f_0)$
both components are functions of the local or particle energy in
the induced potential $U_{0,e}=U_0 + \tilde U_0$. Since the latter
is time-independent, the particle energy is conserved along
particle orbits, i.e., along the characteristics of the Vlasov
equations (\ref{vlasov3d}) and (\ref{vlasov2d}) respectively.
Hence $f_0$ and $\tilde f_0$ satisfy these equations at least 
formally, and we are justified to refer to $(f_0,\tilde f_0)$
as a steady state of the system (\ref{vlasov3d})--(\ref{rhodef}).
We do not discuss the regularity of this steady state further.
However, to conclude this section we want to address the question
whether these states have spatially compact support. 
\begin{prop} \label{compsupp}
Let $(f_{0},\tilde f_{0})$ be a minimizer as obtained 
in Theorem~\ref{existence} and assume that 
\[
0< k< 5/2\ \mbox{and}\ 0 < \tilde k < 1.
\]
Then $U_0,\ \tilde U_0,\ \rho_0,\ \tilde \rho_0\in L^\infty(\R^3)$ with
\[
\lim_{|x|\to \infty} U_0(x) = 0,\ \lim_{|x|\to \infty} \tilde U_0(x) = 0,
\]
$E_0, \tilde E_0 < 0$, and $\rho_0$ and $\tilde \rho_0$ have compact support.
\end{prop}
\begin{proof}
Consider a density $\tilde \rho \in L_+^1\cap L^p(\R^2)$. 
Then $U_{\tilde\rho} \in L^\infty(\R^3)$
with $\lim_{|x|\to \infty} U_{\tilde\rho} (x) = 0$, provided $p>2$. If $\rho$
is defined on $\R^3$ then the same is true provided $p>3/2$. We prove this
assertion for the flat case. Here
\beas
- U_{\tilde\rho} (x) 
&=&
\int_{\R^2} \frac{\tilde\rho(\tilde y)}{|x-(\tilde y,0)|} d\tilde y\\
&=&
 \int_{|x-(\tilde y,0)|\leq R} \frac{\tilde\rho(\tilde y)}{|x-(\tilde y,0)|} d\tilde y
+ \int_{|x-(\tilde y,0)| > R} \frac{\tilde\rho(\tilde y)}{|x-(\tilde y,0)|} d\tilde
  y.\\
&\leq&
 \int_{|\tilde x-\tilde y|\leq \sqrt{R^2-x_3^2}} 
\frac{\tilde\rho(\tilde y)}{|x-(\tilde y,0)|} d\tilde y
+ \frac{||\tilde\rho||_1}{R}\\
&\leq&
C R^{(p-2)/(p-1)} ||\tilde \rho||_{L^p(\{|\tilde x-\tilde y|\leq \sqrt{R^2-x_3^2}\})}
+ \frac{||\tilde\rho||_1}{R}.
\eeas
Since this holds for any $R>0$ the assertion follows; notice that 
for $R>0$ fixed the first term goes to zero for $|x|\to \infty$.
By the weak Young inequality and Lemma~\ref{3dpot2d}, 
$U_0\in L^6(\R^3)$ and $U_0(\cdot,0)\in L^4(\R^2)$, and again
by the weak Young inequality and Lemma~\ref{intflatpot},
$\tilde U_0\in L^4(\R^2)$ and $U_0 \in L^6(\R^3)$. Hence
$U_{0,e} \in L^4(\R^2) \cap L^6(\R^3)$. If we integrate the
relations between $f_0$, $\tilde f_0$, and $U_{0,e}$ from 
Theorem~\ref{minim=ss} with respect to $v$ or $\tilde v$ respectively
we obtain the relations
\be \label{rhourel}
\rho = c (E_0 - U_{0,e})_+^n,\ 
\tilde \rho = \tilde c (\tilde E_0 - U_{0,e}(\cdot,0))_+^{\tilde n},
\ee
where $c$ and $\tilde c$ depend on $\lambda$ and $k$ or 
$\tilde \lambda$ and $\tilde k$ respectively. From the integrability
assertions for the potential we conclude that the spatial
densities have the required integrability provided
$6/n > 3/2$, i.e., $n < 4$ which means $k< 5/2$, and
$4/\tilde n > 2$ i.e., $\tilde n < 2$ which means $\tilde k< 1$.

It therefore remains to show that $E_0 < 0$ and $\tilde E_0 < 0$
as claimed; the assertion on the support of the densities then follows.
Assume that $E_0 > 0$. Then for $|x|$ large, $\rho_0 (x) > c (E_0/2)^n$
which contradicts its integrability, and the same argument works for
$\tilde \rho_0$. Now assume that $E_0 = 0$.
Then $\rho_0 (x) = c (-U_{0,e}(x))^n$, and this again contradicts the 
integrability of $\rho_0$ since $-U_{0,e} \geq C/|x|$ for large $|x|$
and $C>0$. We prove this for $\tilde U_0$, the argument for $U_0$
being completely analogous. We choose $R>0$ such that 
\[
\int_{|\tilde y| \leq R} \tilde \rho_0(\tilde y)\, d\tilde y =: m > 0.
\]
Next we observe that for $|\tilde y| \leq R$ and $|x| \geq 2 R$,
\[
\left|\frac{1}{|x-(\tilde y,0)|} - \frac{1}{|x|}\right| \leq
\frac{R}{(|x|-R)^2}.
\]
If we restrict the convolution integral defining $\tilde U_0$
to the set $\{|y|\leq R\}$ and expand the kernel as indicated
the assertion on $\tilde U_0$ follows.  The same argument works
for $U_0$ so that indeed $-U_{0,e} \geq C/|x|$ as claimed. 
If  $\tilde E_0 = 0$ then  
$\tilde \rho_0 (\tilde x) = c (-U_{0,e}(\tilde x,0))^{\tilde n}$
which contradicts the integrability of $\tilde \rho$.
Notice that under the present assumptions on $k$ and $\tilde k$
it follows that $n<3$ and $\tilde n< 2$.
\end{proof}
\section{Stability} \label{secstab}
\setcounter{equation}{0}
In this section we show how the minimizing property
of a minimizer $(f_0,\tilde f_0) \in \mathcal{F}_{\mathbf{M}}$
leads to a stability estimate. Given a second state
$(f,\tilde f) \in \mathcal{F}_{\mathbf{M}}$ and denoting the
effective potential of the minimizer by $U_{0,e}$ 
a simple computation shows that
\beas 
\mathcal{H} (f,\tilde f)
&=&
\mathcal{H} (f_0,\tilde f_0) + 
\iint \left(\frac{1}{2} |v|^2 + U_{0,e}(x)\right)\, (f-f_0)(x,v)\, dv\, dx \\
&&
{}+ 
\iint \left(\frac{1}{2} |\tilde v|^2 + U_{0,e}(\tilde x,0)\right)\, 
(\tilde f-\tilde f_0)(\tilde x,\tilde v)\, d\tilde v\, d\tilde x \\
&&
{}- ||f-f_0||_\mathrm{pot}^2 
  - ||\tilde f-\tilde f_0||_\mathrm{pot}^2 
-2\langle f-f_0,\tilde f-\tilde f_0\rangle_\mathrm{pot}.
\eeas
With $E(x,v) = \frac{1}{2} |v|^2 + U_{0,e}(x)$ and 
\bea \label{ddef}
d((f,\tilde f),(f_0,\tilde f_0))
&:=& 
\iint E (x,v)\, (f-f_0)(x,v)\, dv\, dx \nonumber \\
&&
{} + \iint E (\tilde x,0,\tilde v,0)\, 
(\tilde f-\tilde f_0)(\tilde x,\tilde v)\, d\tilde v\, d\tilde x
\eea
we can rewrite this expansion as 
\bea \label{hexpansion} 
\mathcal{H} (f,\tilde f)
&=&
\mathcal{H} (f_0,\tilde f_0) + d((f,\tilde f),(f_0,\tilde f_0))
\nonumber\\
&&
{}- ||f-f_0||_\mathrm{pot}^2 
  - ||\tilde f-\tilde f_0||_\mathrm{pot}^2 
-2\langle f-f_0,\tilde f-\tilde f_0\rangle_\mathrm{pot}.
\eea
We need to show that $d((f,\tilde f),(f_0,\tilde f_0))\geq 0$
with equality only if $(f,\tilde f)=(f_0,\tilde f_0)$. To this end
we restrict ourselves to states $(f,\tilde f) \in \mathcal{F}_{\mathbf{M}}$
such that
\be \label{pertrestr}
\int f = \int f_0,\ \int f^{1+1/k} = \int f_0^{1+1/k},\
\int \tilde f = \int \tilde f_0,\ 
\int \tilde f^{1+1/\tilde k} = \int \tilde f_0^{1+1/\tilde k}.
\ee
{\bf Remark.}
From a physics point of view a galaxy and its halo are typically
perturbed by the gravitational field of some distant exterior
object. In particular, the perturbation will result in a measure
preserving redistribution of the particles in phase space, and 
will hence preserve the constraints in (\ref{pertrestr}).
On the other hand, the fact that the perturbations lie 
in $\mathcal{F}_{\mathbf{M}}$ means that the stars are
only shifted within the galactic plane and not perpendicularly to it.
This is certainly an unphysical restriction. To remove
it is a non-trivial problem for future research.

\smallskip

Using (\ref{pertrestr}) and the strict convexity of the
function $[0,\infty[ \ni \zeta \mapsto \zeta^p$ for $p>1$
we find that
\beas
d((f,\tilde f),(f_0,\tilde f_0))
&=&
\iint (E - E_0) \, (f-f_0)
+ \frac{\lambda}{1+1/k} \iint (f^{1+1/k} - f_0^{1+1/k})\\
&&
{} + \iint (E - \tilde E_0)\, (\tilde f-\tilde f_0) 
+ \frac{ \tilde\lambda}{1+1/\tilde k} 
\iint (\tilde f^{1+1/\tilde k} - \tilde f_0^{1+1/\tilde k})\\
&\geq&
\iint
\left[(E - E_0) + \lambda f_0^{1/k}\right]\,(f-f_0)\\
&&
{}+
\iint
\left[(E - \tilde E_0)  +  \tilde\lambda \tilde f_0^{1/\tilde k}
\right]\,(\tilde f-\tilde f_0) \geq 0; 
\eeas
Theorem~\ref{minim=ss} implies that the last expressions are
non-negative, and the strict convexity implies that equality
holds only if $(f,\tilde f) = (f_0,\tilde f_0)$.

In order to establish a stability result we now wish to
apply the above estimates to the time evolution $(f(t),\tilde f(t))$
of a perturbation of $(f_0,\tilde f_0)$. Clearly, we need to 
require that $(f(0),\tilde f(0)) \in \mathcal{F}_{\mathbf{M}}$
satisfies the constraints (\ref{pertrestr}).
More importantly, in view of the fact that nothing is known
on the initial value problem for the system (\ref{vlasov3d})--(\ref{rhodef}),
we have to assume that this system has solutions  
$t \mapsto (f(t),\tilde f(t))$ which preserve the total energy, the constraints
(\ref{pertrestr}), and 
$(f(t),\tilde f(t)) \in \mathcal{F}_{\mathbf{M}}$.
To keep the rest of the discussion simple we furthermore assume that
the minimizer $(f_0,\tilde f_0)$ is unique in $\mathcal{F}_{\mathbf{M}}$
up to spatial shifts. If the minimizer is up to spatial shifts only isolated 
with respect to the distance measurement used in the stability estimate below,
the result remains unchanged. If the minimizers are not even isolated
one can prove the stability of the whole set of minimizers;
we refer to \cite{R7} for the corresponding modifications of
the arguments.

\smallskip

\noindent
{\bf Stability estimate.}
{\em Assume the minimizer $(f_0,\tilde f_0)$ is unique
in $\mathcal{F}_{\mathbf{M}}$. Then for every $\epsilon >0$
there exists $\delta > 0$ such that for any solution
$t \mapsto (f(t),\tilde f(t))$ of (\ref{vlasov3d})--(\ref{rhodef})
satisfying the above assumptions the following is true: If
\[
d((f(0),\tilde f(0)),(f_0,\tilde f_0)) + ||f(0) - f_0||_\mathrm{pot}
+ ||\tilde f(0) - \tilde f_0||_\mathrm{pot} < \delta
\]
then 
\[
d((f(t),\tilde f(t)),(f_0,\tilde f_0)) + ||f(t) - f_0||_\mathrm{pot}
+ ||\tilde f(t) - \tilde f_0||_\mathrm{pot} < \epsilon
\]
up to spatial shifts parallel to the $(x_1,x_2)$ plane and as
long as the solution exists. }

\smallskip

\noindent
We do not call this a theorem because it is not clear that
sufficiently regular solutions to the initial value problem
do indeed exist. Assuming the latter the proof is by contradiction.
If the assertion were false,  there exists $\epsilon >0$,
a sequence of solutions $(t \mapsto (f_j(t),\tilde f_j(t)))$
and a sequence of times $(t_j)$ such that for all
$j\in \mathbb{N}$,
\be \label{initialj}
d((f_j(0),\tilde f_j(0)),(f_0,\tilde f_0)) + ||f_j(0) - f_0||_\mathrm{pot}
+ ||\tilde f_j(0) - \tilde f_0||_\mathrm{pot} < 1/j,
\ee
but
\be \label{finalj}
d((f_j(t_j),\tilde f_j(t_j)),(f_0,\tilde f_0)) + ||f_j(t_j) - f_0||_\mathrm{pot}
+ ||\tilde f_j(t_j) - \tilde f_0||_\mathrm{pot} \geq \epsilon
\ee
regardless of how we shift $(f_j(t_j),\tilde f_j(t_j))$ in space.
Now (\ref{initialj}) and the fact that $d$ is non-negative imply
that all three terms in  (\ref{initialj}) converge to zero.
By Lemma~\ref{toll} this is true also for the mixed term
$\langle f_j(0)-f_0,\tilde f_j(0)-\tilde f_0\rangle_\mathrm{pot}$,
and since $\mathcal{H}$ is preserved, (\ref{hexpansion}) implies
that
\[
\mathcal{H} (f_j(t_j),\tilde f_j(t_j)) = \mathcal{H} (f_j(0),\tilde f_j(0))
\to \mathcal{H} (f_0,\tilde f_0).
\]
Since $(f_j(t_j),\tilde f_j(t_j))\in \mathcal{F}_{\mathbf{M}}$ this means 
that $(f_j(t_j),\tilde f_j(t_j))$ is a minimizing sequence for $\mathcal{H}$
in $\mathcal{F}_{\mathbf{M}}$. By Theorem~\ref{existence} there exists
a subsequence such that up to spatial shifts
\[
||f_j(t_j) - f_0||_\mathrm{pot} + 
||\tilde f_j(t_j) - \tilde f_0||_\mathrm{pot}
\to 0;
\]
at this point the uniqueness assumption for $(f_0,\tilde f_0)$ enters.
Again by Lemma~\ref{toll} this is true also for the mixed term
$\langle f_j(t_j)-f_0,\tilde f_j(t_j)-\tilde f_0\rangle_\mathrm{pot}$,
and (\ref{hexpansion}) with $(f,\tilde f)=(f_j(t_j),\tilde f_j(t_j))$
implies that also  
$d((f_j(t_j),\tilde f_j(t_j)),(f_0,\tilde f_0)) \to 0$.
Hence all three terms in (\ref{finalj}) converge to zero, 
which is a contradiction.

\section{Appendix: Facts on the decoupled minimizers}
\setcounter{equation}{0}
Here we establish the claims on the decoupled variational
problems referred to in Section~\ref{secdecoup}. Several of these
claims, in particular for the non-flat case, can be found in the 
literature.

\smallskip

\noindent
{\em Existence.}\\
In most of the previous investigations the existence of minimizers 
in the decoupled cases was not done exactly for the problems
stated in Section~\ref{secdecoup}: Either 
the Casimir functional, which in the case at hand corresponds to the 
$L^{1+1/k}$ norm, was included into the functional
to be minimized, i.e., an energy-Casimir functional instead
of the energy was minimized, and only the constraint on the $L^1$
norm was posed, or the energy was minimized under the constraint that the
sum of the mass and the Casimir functional was fixed.
In the three dimensional case the problem with two constraints
in the form stated in Section~\ref{secdecoup} has been dealt with
in \cite{aly2,SS}. We briefly show how the method
used in \cite{aly2} can be adapted to the flat case.

With the help of the Riesz rearrangement inequality \cite[3.7]{LL} 
and the fact that the kinetic energy as well as the constraints are invariant 
under spherically symmetric rearrangements the problem is reduced
to minimizing $\mathcal{H}(\tilde f)$ 
where the functions $\tilde f$ in the
constraint set have the form
\[
\tilde f(\tilde x,\tilde v)=\varphi(|\tilde x|,|\tilde v|),
\]
with $\varphi:[0,\infty[^{2}\to [0,\infty[$ non-increasing in each
argument. 
This monotonicity implies that for $1\leq q\leq 1+1/\tilde k$,
\begin{eqnarray*}
\tilde f^{q}(\tilde x,\tilde v)|\tilde x|^{2}|\tilde v|^{2}
&\le& 
C\int_{0}^{|\tilde x|}\int_{0}^{|\tilde
  v|}\varphi^{q}(r,s)rs\,ds\,dr
\le 
C ||\tilde f||_{q}^{q},\\
\tilde f(\tilde x,\tilde v)|\tilde x|^{2}|\tilde v|^{4}
&\le& 
C\int_{0}^{|\tilde x|}\int_{0}^{|\tilde
  v|}\varphi(r,s)rs^{3}\,ds\,dr
\le 
C  E_{\mathrm{kin}}(\tilde f).
\end{eqnarray*}
Hence
\[
\tilde f(\tilde x,\tilde v)\le g(\tilde x,\tilde v):=C \left\{\begin{array}{lll}
|\tilde x|^{-2/q}|\tilde v|^{-2/q},&\mathrm{for}&|\tilde v|\le V(|\tilde x|),\\
|\tilde x|^{-2}|\tilde v|^{-4},&\mathrm{for}&|\tilde v|> V(|\tilde x|),
\end{array}\right.
\]
where $V(|\tilde x|)>0$ is an arbitrary function and the constant $C$
depends on $E_{\mathrm{kin}} (\tilde f)$, $||\tilde f||_{1}$,
and $||\tilde f||_{1+1/\tilde k}$, quantities which are bounded
along minimizing sequences.
The function $g$ induces the spatial density
\begin{eqnarray*}
\rho_{g}(\tilde x)
&=&
C |\tilde x|^{-2/q} \int_{0}^{V(|\tilde x|)} |\tilde v|^{1-2/q}\,
d|\tilde v|
+ C |\tilde x|^{-2}\int_{V(|\tilde x|)}^{\infty} |\tilde v|^{-3}\,
d|\tilde v|\\
&=&
C |\tilde x|^{-2/q} V^{2-2/q}(|\tilde x|) + 
C |\tilde x|^{-2} V^{2}(|\tilde x|).
\end{eqnarray*}
The choice
\[
V(|\tilde x|)=V_{q}(|\tilde x|):=|\tilde x|^{(1-q)/(2q-1)}
\]
yields the estimate
\[
\rho_{g}(\tilde x)\le |\tilde x|^{-2q/(2q-1)}
\]
with the exponent $s:=-2q/(2q-1)$ being such that
\begin{equation} \label{alydichte}
s<-3/2\ \mathrm{for}\ 1<q<3/2,\quad 
s>-3/2\ \mathrm{for}\ q>3/2.
\end{equation}
We split the estimate for $\tilde f$ by choosing $q=1+1/\tilde k>3/2$ 
for $|\tilde x|\le 1$ and $q\in[1,3/2]$ for $|\tilde x|>1$ so that
\[
\tilde f(\tilde x,\tilde v)
\le 
g(\tilde x,\tilde v)
:=
C \left\{\begin{array}{lclcl}
|\tilde x|^{-2/(1+1/\tilde k)}|\tilde v|^{-2/(1+1/\tilde k)}&\mathrm{for}&
|\tilde x|\le 1&\land&|\tilde v|\le V_{1+1/\tilde k}(|\tilde x|),\\
|\tilde x|^{-2}|\tilde v|^{-4}&\mathrm{for}&
|\tilde x|\le 1&\land&|\tilde v|>V_{1+1/\tilde k}(|\tilde x|),\\
|\tilde x|^{-2/q}|\tilde v|^{-2/q}&\mathrm{for}&
|\tilde x|> 1&\land&|\tilde v|\le V_{q}(|\tilde x|),\\
|\tilde x|^{-2}|\tilde v|^{-4}&\mathrm{for}&
|\tilde x|> 1&\land&|\tilde v|> V_{q}(|\tilde x|).
\end{array}\right.
\]
By (\ref{alydichte}),
\[
\rho_{\tilde f}(\tilde x)\le\rho_{g}(\tilde x)\le\left\{\begin{array}{lcl}
Cr^{s_{1}}&\mathrm{with}\ s_{1}>-8/5&\mathrm{for}\ |\tilde x|\le 1,\\
Cr^{s_{2}}&\mathrm{with}\ s_{2}<-8/5&\mathrm{for}\ |\tilde x|> 1.
\end{array}\right.
\]
By the Hardy-Littlewood-Sobolev inequality this implies that
$g$ has finite potential energy.
The crucial step in the existence proof for a minimizer
is the convergence of the potential energy
along a minimizing sequence $(\tilde f_{j})$. 
Since $0\le\tilde f_{j}\le g$, 
the finiteness of the potential energy for $g$ allows us 
to pass to the limit using the dominated convergence theorem. 

\smallskip

\noindent
{\em Saturation of the constraints.}\\
Minimizers of the decoupled problems
always saturate the constraints, i.e., 
$||f_0^{\mathrm{3D}}||_1 = M$, $||f_0^{\mathrm{3D}}||_{1+1/k} = N$,
and similarly for the flat case, because if for a minimizer
one (or both) equalities were replaced by strict inequalities, then 
this minimizer
can be rescaled in such a way that the constraints become
saturated but the energy strictly decreases, which is a contradiction.
A similar argument was used in the proof of Proposition~\ref{saturate}. 

\smallskip

\noindent
{\em The Euler-Lagrange relation and symmetry.}\\
For minimizers of the flat or non-flat problem the phase space
distributions are functions of the local energy, more precisely,
they satisfy relations as
stated in Theorem~\ref{minim=ss}, the only differences being that
\[
E(x,v):=\frac{1}{2}|v|^{2} + U^\mathrm{3D}_{0}(x),\ 
E(\tilde x,\tilde v):=\frac{1}{2}|\tilde v|^{2} + U_{0}^\mathrm{FL}(\tilde x),
\]
and the interaction term $\int U^\mathrm{3D}_0 \rho^\mathrm{FL}_0$ 
in the definitions of $E_0$ and $\tilde E_0$ is dropped.

The asserted symmetry of the minimizers follows
from the fact that symmetric decreasing rearrangements
in $x$ or $\tilde x$ strictly decrease the energy except 
when $\rho^\mathrm{3D}_0$ and $\rho^\mathrm{FL}_0$ and hence also
the induced potentials are symmetric with respect to
some point, cf.~\cite[3.7, 3.9]{LL}.

\smallskip

\noindent
{\em Virial identity and compact support.}\\
Both flat and non-flat minimizers satisfy the virial
identity
\[
2 E_\mathrm{kin}(f_0^{\mathrm{FL}})=E_\mathrm{pot}(f_0^{\mathrm{FL}}),\
2 E_\mathrm{kin}(f_0^{\mathrm{3D}})=E_\mathrm{pot}(f_0^{\mathrm{3D}}).
\]
This follows from the fact that these minimizers are time-independent
solutions of the corresponding Vlasov-Poisson systems. A direct proof
based on their minimizing property and scaling is given in \cite[3.2]{aly2}.
The virial identities together with the restrictions on
$k$ and $\tilde k$ immediately imply that the cut-off energies
$E_0$ and $\tilde E_0$ in the Euler-Lagrange relations are
strictly negative. In order to show that the minimizers have
compact support it therefore suffices to show that their 
potentials converge to zero at spatial infinity. We indicate the
corresponding argument for the 3D case, the flat one being completely
analogous. Applying the H\"older inequality to the first term
in the estimate
\[
-U_0^\mathrm{3D}(x) \leq 
\int_{|x-y|\leq R} \frac{\rho_0^\mathrm{3D}(y)}{|x-y|} dy + \frac{M}{R},\
R>0,
\]
implies that the potential is bounded and vanishes at spatial infinity,
provided $\rho_0^\mathrm{3D} \in L^p(\R^3)$ with $p>3/2$. While 
a-priori this need not be the case for $0<k<7/2$ we can use the fact 
that similar to (\ref{rhourel}),
\[
\rho_0^\mathrm{3D} = c (E_0 - U_0^\mathrm{3D})_+^n,
\]
start with the known integrability $U_0^\mathrm{3D} \in L^6(\R^3)$
to conclude that $\rho_0^\mathrm{3D} \in L^{6/n}(\R^3)$, and obtain a 
new integrability estimate for the potential through the weak Young
inequality. After finitely many iterations the desired integrability
of $\rho_0^\mathrm{3D}$ follows, cf.\ \cite[Prop.~2.7]{rein07}.

\smallskip

\noindent
{\em Uniqueness for the 3D problem.}\\
First we notice that by the virial relation the Lagrange
multipliers are uniquely determined by the constraint
parameters $M$ and $N$.
In the 3D case we can now continue as follows. A minimizer
is completely determined by its potential. The latter satisfies
the Emden-Fowler equation
\[
\frac{1}{r^2}(r^2 U')' = c (E_0 - U)_+^n,\ \mbox{i.e.},\
\frac{1}{r^2}(r^2 y')' = -c y_+^n,
\]
where $y= E_0 - U$, and the constant $c$ depends only on $k$ and 
$M,\, N$.
The solutions of this equation which are regular at the origin
are uniquely determined by their value at the origin. Moreover,
the scaling $y_\alpha (r) = \alpha y(\alpha^\lambda r)$ with 
$\lambda = (k+1/2)/2$ turns solutions into solutions.
But the mass constraint fixes this scaling, and uniqueness
of the minimizer follows. For more details we refer to 
\cite[p.~464]{rein07}. Unfortunately, there is no analogue
to the Emden-Fowler equation in the flat case---the flat potential
does not satisfy the Poisson equation---and uniqueness in the flat
case is not known.

\smallskip

\noindent
{\em The radius relation (\ref{radius3d}) in the 3D case.}\\
Each minimizer
is a spherically symmetric steady state $(f,\rho,U)$
of the three dimensional Vlasov-Poisson system,
and for each choice of $M$ and $N$ there is a unique 
such steady state.
If $(f,\rho,U)$ is a steady
state and $\alpha , \beta >0$ are arbitrary, then 
\[
f_{\alpha \beta}(x,v) = \alpha^2 \beta f(\alpha x, \beta v),\
\rho_{\alpha \beta}(x) = \alpha^2 \beta^{-2} \rho(\alpha x),\
U_{\alpha \beta}(x) = \beta^{-2} U(\alpha x)
\]
defines another one. 
There is a unique steady state $(f^\ast,\rho^\ast,U^\ast)$
with $||f^\ast||_1 = 1 = ||f^\ast||_{1+1/k}$, and the 
minimizer with general $M$ and $N$ is obtained by rescaling
$f^\ast$ with the parameters 
\[
\alpha = M^{(1-2k)/3} N^{2k+2)/3},\ \beta = M^{(k-2)/3} N^{-(k+1)/3}.
\]
Since $R = R^\ast / \alpha$ this implies (\ref{radius3d}). 

\smallskip

\noindent
{\em The radius relation (\ref{radius2d}) in the flat case.}\\
We do not know whether for each choice of $M$ and $N$
there exists a unique minimizer $f_0^\mathrm{FL}$,
and so we cannot use the argument above to prove (\ref{radius2d}).

To obtain a two-parameter family of minimizers which obeys 
the radius relation (\ref{radius2d}) we proceed as follows.
Since minimizers a-posteriori saturate the constraints we redefine
\[
\mathcal{F}_{M,N}^\mathrm{FL} := 
\Bigl\{ \tilde f\in L^{1}_{+}(\mathbb{R}^{4}) \mid
||\tilde f||_{1} = M,\; ||\tilde f||_{1+1/\tilde k} =N,\
E_{\mathrm{kin}}(\tilde f)<\infty \Bigr\}.
\]
For $\tilde f \in \mathcal{F}^\mathrm{FL}_{1,1}$
we define the rescaled function
\[
\tilde f_{\mu,\nu} (\tilde x,\tilde v) := 
\mu \tilde f(\mu \tilde x,\nu \tilde v).
\]
Then
\[
|| \tilde f_{\mu,\nu}||_1 = \mu^{-1} \nu^{-2},\
|| \tilde f_{\mu ,nu}||_{1+1/\tilde k} = 
\mu^{(1-\tilde k)/(1+\tilde k)} \nu^{-2 \tilde k/(\tilde k +1)},\
\mathcal{H}(\tilde f_{\mu,\nu}) = \mu^{-1} \nu^{-4}
\mathcal{H}(\tilde f).
\]
For $M,\, N>0$ arbitrary we choose $\mu,\,\nu$ such that 
$\tilde f_{\mu,\nu} \in \mathcal{F}_{M,N}^\mathrm{FL}$,
i.e.,
\[
\mu = M^{-\tilde k} N^{\tilde k +1},\ 
\nu = M^{(\tilde k -1)/2} N^{-(\tilde k+1)/2}.
\]
The mapping 
$\mathcal{F}^\mathrm{FL}_{1,1} \ni \tilde f \mapsto \tilde f_{\mu\nu}
\in \mathcal{F}_{M,N}^\mathrm{FL}$ is one-to-one and onto.
Let $f^\mathrm{FL}$ 
denote an arbitrary, spherically symmetric 
minimizer of $\mathcal{H}$ 
over the set $\mathcal{F}_{1,1}^\mathrm{FL}$. 
Then $f_{\mu,\nu}^\mathrm{FL}$ is a minimizer
of  $\mathcal{H}$ over  $\mathcal{F}_{M,N}^\mathrm{FL}$
which we denote by $f_{M,N}^\mathrm{FL}$. To see
this we observe that any function 
$\tilde g \in \mathcal{F}_{M,N}^\mathrm{FL}$
can be written as $\tilde g=\tilde f_{\mu,\nu}$ where 
$\tilde f\in \mathcal{F}^\mathrm{FL}_{1,1}$, and hence
\[
\mathcal{H}(\tilde g) = \mathcal{H}(\tilde f_{\mu,\nu})
= \mu^{-1} \nu^{-4} \mathcal{H}(\tilde f) 
\geq \mu^{-1} \nu^{-4} \mathcal{H}(f^\mathrm{FL})
= \mathcal{H}(f_{\mu,\nu}^\mathrm{FL}).
\]
If $\tilde R^\ast$
denotes the radius of the spatial support of $f^\mathrm{FL}$,
then the spatial support of $f_{M,N}^\mathrm{FL}$
has radius
\[
\tilde R = \mu^{-1} \tilde R^\ast = 
\tilde R^\ast M^{-\tilde k} N^{\tilde k +1},
\]
and this is the remaining assertion (\ref{radius2d}).

\end{document}